%% file: hgraph.tex
\documentclass[a4paper,english,cleveref, autoref, thm-restate]{lipics-v2019}

\title{Recognizing Proper Tree-Graphs} %

\titlerunning{Recognizing Proper Tree-Graphs} %

\author{Steven Chaplick}{Maastricht University, The Netherlands}{s.chaplick@maastrichtuniversity.nl}{https://orcid.org/0000-0003-3501-4608
}{}%

\author{Petr A.~Golovach}{Department of Informatics, University of Bergen, Bergen, Norway}{petr.golovach@uib.no}{https://orcid.org/0000-0002-2619-2990}{Author supported by the project MULTIVAL of the Research Council of Norway.}

\author{Tim A.~Hartmann}{RWTH Aachen, Germany}{hartmann@algo.rwth-aachen.de}{https://orcid.org/0000-0002-1028-6351}{}

\author{Du\v{s}an Knop}{Department of Theoretical Computer Science, Faculty of Information Technology,\\ Czech Technical University in Prague, Prague, Czech Republic}{dusan.knop@fit.cvut.cz}{https://orcid.org/0000-0003-2588-5709}{Supported by the OP VVV MEYS funded project CZ.02.1.01/0.0/0.0/16\_019/0000765 ``Research Center for Informatics''.}

\authorrunning{Steven Chaplick, Petr A.~Golovach, Tim A.~Hartmann, Du\v{s}an Knop} %

\Copyright{Steven Chaplick, Petr A.~Golovach, Tim A.~Hartmann, Du\v{s}an Knop} %



\relatedversion{} %

\supplement{}%

\acknowledgements{We thank Peter Zeman for initial discussions on this work, particularly relating to the \cNP-completeness of proper $H$-graph recognition.
We also thank an anonymous reviewer for pointing out a mistake in an earlier version.}%

\nolinenumbers %

\hideLIPIcs  %

\EventEditors{Yixin Cao and Marcin Pilipczuk}
\EventNoEds{2}
\EventLongTitle{15th International Symposium on Parameterized and Exact Computation (IPEC 2020)}
\EventShortTitle{IPEC 2020}
\EventAcronym{IPEC}
\EventYear{2020}
\EventDate{December 14--18, 2020}
\EventLocation{Hong Kong, China (Virtual Conference)}
\EventLogo{}
\SeriesVolume{180}
\ArticleNo{8}
\input{header-hgraphs}

\begin{document}

\maketitle

\begin{abstract}
We investigate the parameterized complexity of the recognition problem for the proper $H$-graphs.
The $H$-graphs are the intersection graphs of connected subgraphs of a subdivision of a multigraph $H$, and the properness means that the containment relationship between the representations of the vertices is forbidden.
The class of $H$-graphs was introduced as a natural (parameterized) generalization of interval and circular-arc graphs by Bir\'o, Hujter, and  Tuza in 1992, and the proper $H$-graphs were introduced by Chaplick et al. in WADS 2019 as a generalization of proper interval and circular-arc graphs. For these graph classes, $H$ may be seen as a structural parameter reflecting the distance of a graph to a (proper) interval graph,
	and as such gained attention as a structural parameter in the design of efficient algorithms.
We show the following results.
\begin{itemize}
\item For a tree $T$ with $t$ nodes, it can be decided in \( 2^{\Oh(t^2 \log t)} \cdot n^3 \) time, whether  an $n$-vertex graph \( G \)  is a proper \( T \)-graph.
For yes-instances, our algorithm outputs a proper $T$-representation.
This proves that the recognition problem for proper $H$-graphs, where $H$ required to be a tree, is fixed-parameter tractable when parameterized by the size of $T$.
Previously only \cNP-completeness was known.
\item  Contrasting to the first result, we prove that if $H$ is not constrained to be a tree, then the recognition problem becomes much harder. Namely, we show that there is a multigraph $H$ with 4 vertices and 5 edges such that it is \cNP-complete to decide whether $G$ is a proper $H$-graph.
\end{itemize}
\end{abstract}

\newpage

\newcommand{\repr}{\mathcal{M}}

\section{Introduction}

An \emph{intersection representation} of a graph $G=(V,E)$ is a collection of nonempty sets $\set{M_v}{v \in V(G)}$ over a given universe such that $\{u,v\}$ is an edge of $G$ if and only if $M_u \cap M_v \neq \emptyset$.
A large area of research in graph algorithms is the study of restricted families of graphs arising from specialized intersection representations, e.g., the \emph{interval} graphs are the graphs with an intersection representation where the sets are intervals of $\mathbb{R}$, and the \emph{circular-arc} graphs are intersection graphs of families of arcs of the circle.
The interval graphs and similarly defined graph classes are often motivated from application areas such as circuit layout problems~\cite{sinden1966,BradyS1990}, scheduling problems~\cite{roberts1978graph}, biological problems~\cite{JosephMT92}, or the study of wireless networks~\cite{HusonS1995-wireless}.
We refer to the books~\cite{BrandstadtLS99,Golumbic04} for an introduction and survey of the known results on the related graph~classes.

A key feature of these specialized intersection representations is that they can often be used to obtain efficient algorithms for standard combinatorial optimization problems, e.g., it is well-known~\cite{BrandstadtLS99} that the \textsc{Clique} and \textsc{Independent Set} problems, as well as various coloring and Hamiltonicity problems are all efficiently solvable on interval graphs, and the algorithms often leverage on the intersection representation.
This led Bir\'o~et~al.~\cite{biro1992precoloring} to introduce an elegant %
 family of intersection graph classes, called \emph{$H$-graphs},  over universes that may be seen as (multi) graphs.
Formally, the parameter $H$ is a multigraph, and a graph $G$ is an \define{$H$-graph} when there is a
\define{subdivision} $\Hsub$ of $H$\footnote{%
$\Hsub$ is obtained from \(H\) by iteratively replacing an edge \(\{u,v\}\) by a path \(uwv\), where \(w\) is a new vertex.}  and a collection \( \repr = \set{ M_v \subseteq V(\Hsub) }{ v \in V(G) } \) of sets, where we refer to $M_v$ as the \define{model of $v$}, such that
\begin{itemize}
	\item
	for every $v \in V(G)$, its model $M_v$ induces a connected subgraph of \( \Hsub \), and
	\item
	\( \{u,v\} \in E(G) \) if and only if \( M_u \cap M_v \neq \emptyset \).
\end{itemize}
In this context, $\Hsub$ \define{represents} $G$.
Observe that, for any interval graph $G$, there is a path $P$ (i.e., a subdivision of $K_2$) such that $P$ represents $G$ meaning that  the interval graphs are precisely the $K_2$-graphs.
Similarly, the circular-arc graphs are $C$-graphs for any cycle $C$, and every chordal graph is a $T$-graph for some tree $T$, i.e., indeed, $H$-graphs can be seen as a parameterized generalization of several important families of intersection graphs, where $H$ is a parameter reflecting the distance of a graph to an interval graph.
Bir\'o~et~al.~\cite{biro1992precoloring} provided polynomial-time algorithms (via treewidth-based techniques) for coloring problems on $H$-graphs for fixed $H$, but left many interesting problems open.

The classes of $H$-graphs have seen renewed interest in recent years concerning their structure and recognition~\cite{ChaplickTVZ17}, relation to other graph parameters~\cite{ChaplickTVZ17,FominGR2020}, and primarily regarding the computational complexity of standard algorithmic problems when parameterized by the size of $H$~\cite{aaolu2019isomorphism,ChaplickFGK019,ChaplickTVZ17,ChaplickZ17,%
FominGR2020,JaffkeKST19,JaffkeKT20}.
Of particular relevance to our paper is the work on Hamiltonicity problems~\cite{ChaplickFGK019} as it introduces proper $H$-graphs, which are to \emph{proper interval graphs} %
as $H$-graphs are to interval graphs.
Namely,
for a graph $G$, a subdivision $\Hsub$ of $H$ \define {properly represents} $G$ when $\Hsub$ represents $G$ using models $\set{M_v \subseteq V(\Hsub)}{v \in V(G)}$ such that for each $u,v \in V(G)$, neither \( M_u \subseteq M_v \) nor \( M_v \subseteq M_u \).
In particular, on proper $H$-graphs polynomial size kernels (in the size of $H$) were developed for various Hamiltonicity problems~\cite{ChaplickFGK019}, but the recognition problems were left open.

The cornerstone problem for every graph class is recognizability, and we focus on the recognition problem for proper $H$-graphs both when $H$ is part of the input and when $H$ is fixed.
It is important to note that the problem of testing whether for a given graph $G$ and given tree $T$, the graph $G$ is a $T$-graph is NP-complete~\cite[Theorem~4]{KKOS15}.
In fact, the reduction~\cite[Theorem~4]{KKOS15} also implies that testing whether $G$ is a \emph{proper} $T$-graph is NP-complete.
In contrast to this, it is known that when $T$ is fixed, testing whether a given graph is a $T$-graph can be done in polynomial-time~\cite{ChaplickTVZ17}, i.e., \cXP in the size of $T$; but it is not known whether the problem is \cFPT in the size of $T$.
When going beyond trees the recognition problem becomes much harder.
Namely, for each fixed non-cactus graph $H$, $H$-graph recognition is \cNP-complete~\cite{ChaplickTVZ17}.
However, for fixed $H$, it seems that only two cases of proper $H$-graph recognition have been studied:
The proper interval graphs (proper $K_2$-graphs)~\cite{Corneil04,DengHH1996} and the proper circular-arc graphs (proper $C$-graphs, for any cycle $C$)~\cite{DengHH1996} can each be recognized in linear-time.

\subparagraph{Our Contribution.}
In our main result, we show that the recognition of proper $T$-graphs is fixed-parameter tractable (\cFPT) with respect to the size of $T$ by proving the following.

\begin{restatable}{theorem}{algorithm}
\label{lemma:algorithm}
There is an algorithm that, given an $n$-vertex graph \( G \) and a tree $T$ with $t$ nodes, decides whether \( G \) is a proper \( T \)-graph,
and if yes, outputs a proper $T$-representation, in  \( 2^{\Oh(t^2 \log t)} \cdot n^3 \) time.
\end{restatable}

To obtain our \cFPT algorithm for proper $T$-graph recognition, we first observe that the problem can be reduced to the case when the input graph $G$ is connected and chordal.
We proceed in the following three key steps.

In Section~\ref{sec:compact-reps}, we introduce compact representations which are an analog to the clique-trees of chordal graphs that incorporates the properness condition.
We characterize the proper $T$-graphs via these compact representations. This allows us to work with maximal cliques of the input graph that can be listed in linear-time due to the chordalilty of $G$.

In Subsections~\ref{subsection:surround:nodes} and \ref{subsection:chains}, independent of the tree $T$, we partition the maximal cliques into a collection of the so-called \emph{chains} each one necessarily forming a path in any proper $T$-representation, and the remaining singleton cliques that are marked and treated separately.
We show that having a compact $T$-representation means there are, in terms of the size of $T$, at most quadratically many of these marked cliques  and chains altogether.

In Subsections~\ref{sec:template} and \ref{sec:norm}, we combine these ideas to form our \cFPT algorithm for proper $T$-graph recognition.
First, our algorithm guesses a
layout of the chains and the marked maximal cliques. The remaining non-trivial task is to decide whether there is a compact representation corresponding to the guessed layout. We select a root of the tree and show a combinatorial result that if any  compact representation realizes some  layout, it can be assumed to have some special properties concerning the  usage of the nodes of degree at least three of the tree by the models with respect to the root.  We call representations satisfying these properties  \emph{normalized}.
Our algorithm follows the layout bottom-up
and constructs  a normalized representation if it exists.

\medskip
We complement our algorithmic result from Theorem~\ref{lemma:algorithm} by proving that if $H$ is not constrained to be a tree, the recognition problem for proper $H$-graphs becomes \cNP-complete even if $H$ has bounded size.
This negative result employs a reduction quite similar to the one used for (non-proper) $H$-graphs in~\cite{ChaplickTVZ17}, and 
is discussed and proven in Section~\ref{sec:rec-hard}
\begin{restatable}{theorem}{recHard}
\label{thm:rec-hard-intro}
For the 4-vertex, 5-edge multigraph $\Dplus$ defined by
$V(\Dplus)= \{a,b,c,d\}$ and $E(\Dplus) = \{ab, bc, bc, bc, cd\}$ (see also Figure~\ref{fig:diamond-plus}(a)\sv{ in the appendix}), proper $\Dplus$-graph recognition is \cNP-complete.
\end{restatable}
\input{preliminaries}

\input{proper}

\input{proper-tree-graphs}

\input{recognition-hardness}

\section{Concluding Remarks and Open Problems}
\label{sec:conclusion}

Our recognition algorithm for proper tree-graphs provides the following side result on proper leafage (introduced by Lin et al.~\cite{DBLP:journals/dmgt/LinMW98} analogously to leafage):
The \define{proper leafage} $\ell^\star$ of a chordal graph $G$ is the minimum number of leaf nodes of all trees $T$ that properly represent $G$.
The side result, as in \cref{cor:pleafage}, is that computing the proper leafage is \cFPT.
For the decision version, if $G$ is not a proper interval graph, we simply guess the host tree $\Tsub$ of minimal leafage and verify properness with our algorithm from \autoref{lemma:algorithm}. Of course, it still remains open whether computing proper leafage is \cNP-hard.

\begin{cor}\label{cor:pleafage}
Computing the proper leafage $\ell^\star$ of a chordal graph $G$ is \cFPT\ w.r.t.~$\ell^\star$.
\end{cor}

While we have shown that proper $T$-graph recognition is \cFPT, it remains open whether non-proper $T$-graph recognition is \cFPT.
Perhaps most importantly, gaps remain concerning the precise conditions under which (proper) $H$-graph recognition is \cNP-complete for fixed $H$.

\bibliography{lit}

\newpage
\appendix

\sv{\section{Omitted Proofs}}

\appendixText

\end{document}

%% file: header-hgraphs.tex
\usepackage{caption}
\usepackage{subcaption}

\usepackage{xcolor}

\usepackage{apxproof} 

\usepackage{hyperref}
\usepackage[all]{hypcap}

\usepackage{tikz}
\usetikzlibrary{shapes}

\usepackage{amsmath}
\usepackage{amsthm}
\usepackage{mathtools} 
\usepackage{xspace} 
\usepackage{tabularx} 

\usepackage[utf8]{inputenc}
\usepackage[textsize=tiny,textwidth=2cm]{todonotes}

\newcommand{\CC}{\mathcal{C}}

\newcommand{\problem}[1]{\textsc{#1}\xspace}

\newcommand{\cNP}{\hbox{\textsf{NP}}\xspace}

\newcommand{\cXP}{\hbox{\textsf{XP}}\xspace}
\newcommand{\cFPT}{\hbox{\textsf{FPT}}\xspace}

\newcommand{\hid}{\problem{IntDim(1,3)}\xspace}
\newcommand{\Dplus}{\mathfrak{D}\xspace}
\newcommand{\calP}{\mathcal{P}\xspace}

\theoremstyle{plain}
\newtheoremrep{cor}[theorem]{Corollary}
\newtheoremrep{observation}[theorem]{Observation}

\newcommand{\Oh}{\mathcal{O}}

\newcommand{\define}[1]{\emph{#1}}
\newcommand{\stress}[1]{\textit{#1}}

\newcommand{\backward}{(\( \Leftarrow \))\xspace}
\newcommand{\forward}{(\( \Rightarrow \))\xspace}

\newcommand{\set}[2]{ \{ #1 \; | \; #2 \} }

\makeatletter
\newcommand{\labeltext}[3][]{%
    \@bsphack%
    \csname phantomsection\endcsname
    \def\tst{#1}%
    \def\labelmarkup{}
    \def\refmarkup{}%
    \ifx\tst\empty\def\@currentlabel{\refmarkup{#2}}{\label{#3}}%
    \else\def\@currentlabel{\refmarkup{#1}}{\label{#3}}\fi%
    \@esphack%
    \labelmarkup{#2}
}
\makeatother

\newcommand{\V}[2][{}]{V_{#2}^{#1}}
\newcommand{\y}{{{y}}}
\newcommand{\z}{{{z}}}

\newcommand{\subdividedGraph}[1]{#1_{\operatorname{sub}}}
\newcommand{\Hsub}{\subdividedGraph{H}}
\newcommand{\Tsub}{\subdividedGraph{T}}

\newcommand{\YY}{\mathcal{Y}}

\newcommand{\eye}{eye\xspace} 
\newcommand{\eyes}{eyes\xspace}

\newcommand{\KK}[1]{\mathbf{\Gamma}(#1)}
\newcommand{\KKI}[1]{\mathbf{\Gamma}^{-1}(#1)}

\newcommand{\guard}{guard\xspace}
\newcommand{\guards}{guards\xspace}

\newcommand{\Ttemp}{T^0}
\newcommand{\ttemp}{t^0}
\newcommand{\htemp}{h^0}

\newcommand{\ty}[1]{{\lambda}_{#1}}
\newcommand{\typ}[1]{{\lambda'_{#1}}}

\input{appHandler}

\newcommand{\lcolor}{red!30}
\newcommand{\ycolor}{green!37}
\newcommand{\rcolor}{blue!30}
\newcommand{\zcolor}{black!14}
\newcommand{\itips}{0.08}
\newcommand{\intervallca}[4]{
	\draw (#1,#3+\itips)--(#1,#3-\itips) [thick,draw=#4];
	\draw (#1,#3) -- (#2,#3) [thick,draw=#4]; 
	}
\newcommand{\intervallc}[4]{
	\draw (#1,#3+\itips)--(#1,#3-\itips) [thick,draw=#4];
	\draw (#1,#3) -- (#2,#3) [thick,draw=#4]; 
	\draw (#2,#3+\itips)--(#2,#3-\itips) [thick,draw=#4];
	}
\newcommand{\intervall}[3]{
	\intervallc{#1}{#2}{#3}{black}
	}
\newcommand{\vintervallca}[4]{
	\draw (#3+\itips,#1)--(#3-\itips,#1) [thick,draw=#4];
	\draw (#3,#1) -- (#3,#2) [thick,draw=#4]; 
}
\newcommand{\vintervallc}[4]{
	\draw (#3+\itips,#1)--(#3-\itips,#1) [thick,draw=#4];
	\draw (#3,#1) -- (#3,#2) [thick,draw=#4]; 
	\draw (#3+\itips,#2)--(#3-\itips,#2) [thick,draw=#4];
}
\newcommand{\vintervall}[3]{
	\vintervallc{#1}{#2}{#3}{black}
}
\newcommand{\intervalll}[3]{
	\intervall{#1+\igap}{#2-\igap}{#3}
	}
\newcommand{\igap}{0.1}
\newcommand{\iiclique}[3]{
	\node[ellipse,draw=#1,fill=#1,label=above:#3,minimum height=20] at (#2,1.2) {};
}

%% file: appHandler.tex
\usepackage{etoolbox}

\newcommand{\sv}[1]{}
\newcommand{\lv}[1]{#1}

\newcommand{\appendixText}{}

\renewcommand{\toappendix}[1]{#1}

%% file: preliminaries.tex
\section{Preliminaries}

\subparagraph{General Notation.}
We consider undirected graphs $G$ with vertex set $V(G)$ and edge set $E(G)$.
Usually we denote an edge as a set $\{u,v\}$.
However, when needed, we also denote an edge an ordered pair $(u,v)$.
For any subset $W$ of $V(G)$, we use $N(W)$ to denote the \emph{open neighborhood} of $W$, i.e., $N(W) \coloneqq \{ u \in V(G) \setminus W \mid \{u,w\} \in E(G), w \in W\}$, and for a single vertex $w \in V(G)$, $N(w)\coloneqq N(\{w\})$.
We denote the set of maximal cliques of a graph as $\CC(G)$.
The shorthand $[n]$ denotes the set $\{1,\dots,n\}$ of integers.

A \define{subdivision} $H'$ of a graph $H$ at an edge $\{u,w\}$ is the graph resulting from replacing edge $\{u,w\}$ with a path $u,v,w$ where $v$ is new vertex.
A \define{contraction} of a graph $H$ at an edge $\{u,w\}$ is the graph resulting from removing edge $\{u,v\}$ and identifying the two vertices $u$ and $w$.
Then $\Hsub$ is a \define{re-subdivision} of $H$ if it can be obtained by a series of contractions of $H$ (possibly none) followed by a series of subdivisions.
In particular a graph $G$ is a (proper) $H$-graph if and only if there is a re-subdivision that (properly) represents $G$.

Let \( T \) be a tree.
For any pair \( x,y \) of nodes of \(T\), we denote by \( T[x,y] \) the set of nodes of the unique path from \( x \) to \( y \) in \(T\).
Note that \( T[x,y] = T[y,x] \).
We similarly define \( T(x,y] := T[x,y] \setminus \{x\} \) and \( T(x,y) := T[x,y] \setminus \{x,y\} \).
A tuple of nodes \( (x_1,\dots,x_s) \) is \define{\( T \)-ordered} if there exists a path in the graph \( T \) from \( x_1 \) to \( x_s \) where the nodes \( x_1,\dots,x_s \) occur in this order, i.e., \(T[x_1,x_s]\) is the path \(x_1,\ldots,x_s\).

\subparagraph{$H$-graphs.}
Consider a re-subdivision $\Hsub$ of a graph $H$ that (properly) represents a graph $G$ using models $\set{M_v \subseteq V(\Hsub)}{v \in V(G)}$.
For clarity, we refer to each \( x \in V(\Hsub) \) as a \define{node} and to each \( v \in V(G) \) as a \define{vertex}.
We further refer to each node \( x \in V(\Hsub) \) as:
\begin{itemize}
	\item a \define{subdivision node} when it has degree two,
	\item a \define{branching node} when it has degree more than two, and
	\item a \define{leaf node} if it has degree one.
\end{itemize}
For a set of nodes \( X \subseteq V(\Hsub) \), let \( \V{X} := \set{ v \in V(G) }{ M_v \cap X \neq \emptyset } \).
When \( X = \{x\} \), we also write \( \V{x} \) to mean \( \V{\{x\}} \).
For a subset of vertices \( \Gamma \subseteq V(G) \), let \( M_\Gamma := \bigcup_{v \in \Gamma} M_v \).
We say that a set \( \Gamma \) of vertices (or nodes) is \define{connected} if the graph induced by~\( \Gamma \) is connected.

\begin{observation}
	\label{lemma:M:K:is:connected}
	Let $\Hsub$ (properly) represent a graph $G$.
	For any connected subset $\Gamma$ of $V(G)$, the model $M_\Gamma$ of $\Gamma$ is connected in $\Hsub$.
\end{observation}
\toappendix{
\lv{\begin{proof}}\sv{\begin{proof}[Proof of \cref{lemma:M:K:is:connected}]}
	The proof is by induction on \( |\Gamma| \).
	If \( |\Gamma|=1 \), then we are done.
	Now suppose \( |\Gamma| \ge 2 \).
	Since \( G[\Gamma] \) is connected, it has a spanning tree which contains a leaf~\( v \) and let \( u \) be its neighbor in the assumed spanning tree.
	Thus, we can apply induction on \( \Gamma \setminus \{v\} \) and obtain that \( M_{\Gamma \setminus \{v\}} \) is connected in \( \Hsub \).
	Recall that since \( uv \) is an edge in \( G[\Gamma] \), we have \( M_u \cap M_v \neq \emptyset \) by definition.
	Consequently, \( M_\Gamma \) is connected and we are done.
\end{proof}
}

\subparagraph{Chordal Graphs and Clique Trees.}
A graph is \emph{chordal} when it does not contain an induced $k$-vertex cycle for any $k \geq 4$.
The chordal graphs are well known to be characterized as the intersection graphs of subtrees of a tree, i.e., for every chordal graph $G$, there is a tree $T$ that represents $G$ ($G$ is a $T$-graph)~\cite{Buneman74,Gavril1974,Walter1978}.
In fact, $G$ is chordal if and only if there is a tree $T$ with models $\set{M_v \subseteq V(T)}{v \in V(G)}$ where, for each node $x \in V(T)$, $V_x$ is a maximal clique of $G$ and for every node $y \in V(T)$ with $y \neq x$, $V_y \neq V_x$~\cite{Buneman74,Gavril1974,Walter1978}.
These special representations of $G$ are called \define{clique trees}, and one can be constructed in linear-time~\cite{BlairP1993,GalinierHP1995}.
Note that chordal graphs have a simpler linear-time recognition algorithm~\cite{RoseTL1976}.
Finally, every chordal graph $G$ has at most $n$ maximal cliques where $n=|V(G)|$ and the sum of the sizes of the maximal cliques of $G$ is $O(n+m)$~\cite{Golumbic04}.
In particular, the total size of a clique tree of $G$ is $O(n+m)$.
Clearly, the latter two properties of chordal graphs also apply to (proper) $T$-graphs independently of $T$, and we will use them implicitly throughout our discussions.

Each chordal graph $G$ is also a proper $T$-graph for a tree~$T$.
Namely, if a tree~$T$ represents $G$ via models $\set{M_v}{v \in V(G)}$, any tree $T'$ built from $T$ as follows properly represents $G$: Extend each model $M_v$ by a new node $x_v$ and add $\{x,x_v\}$ to $E(T)$ for some $x \in M_v$.%

%% file: proper.tex
\section{Compact Representations of Proper T-Graphs}
\label{sec:compact-reps}

\newcommand{\nn}{\mathbb{N}}

\newcommand{\compact}{compact\xspace}
\newcommand{\compactness}{compactness\xspace}
\newcommand{\compactly}{compactly\xspace}

\newcommand{\escape}{escape\xspace}
\newcommand{\escapes}{escapes\xspace}

In this section we introduce an analogue of clique trees for proper \(T\)-graphs.
Ideally, $G$ being a proper \(T\)-graph would imply a clique tree with the topology of $T$ representing $G$ which satisfies properness; in other words:
	a re-subdivision $\Tsub$ of $T$ with models satisfying properness (i.e., forbidding \( M_u \subseteq M_v \) for every pair \(u,v \in V(G) \)) such that every node $x$ represents a unique maximal clique $\V{x}$.
However, a proper tree-representation of a graph \( G \) may use a lot of nodes just to ensure that the models \( M_u \) and \( M_v \) obey properness; which is already the case for $K_2$ and its interval representation.
Fortunately we may guarantee that almost all nodes represent a unique maximal clique by relaxing the properness condition.
Instead of forbidding containment, we require that when $M_u $ intersects $M_v$, there is a \emph{place} where \( M_u \) may be extended (as needed) to break containment.
That place is an edge \( \{x,y\} \) in the tree \(\Tsub\) where $u$ \define{strongly escapes} $v$, that is, $u,v \in \V{x}$ and $v \notin \V{\y}$.
Actually, a weaker version of \define{escape} suffices.
A vertex $u$ \define{escapes} $v$ if $u \in \V{x} \) and \( v \notin \V{y}$.

\begin{definition}
Let a tree $\Tsub$ with models $\{M_u\mid u\in V(G)\}$ represent a connected graph~$G$.
We say that $\Tsub$ is a \define{\compact representation of~$G$} if
	\begin{enumerate} %
		\item[\labeltext{(C1)}{it:compactRep:emptyLeaves}]
		for every leaf node \( x \in V(\Tsub) \), \( \V{x} = \emptyset \),
		\item[\labeltext{(C2)}{it:compactRep:uniqueMaximal}]
		there is a bijection between the non leaves of $V(\Tsub)$ and the maximal cliques $\CC(G)$, and
		\item[\labeltext{(C3)}{it:compactRep:escapes}] for every ordered pair \( (u,v) \) with \( u,v \in V(G) \),
			there is an edge \mbox{\( \{x,y\} \in E(\Tsub) \)} where \( u \) \escapes \( v \).
	\end{enumerate}
\end{definition}

\begin{observation}\sv{[$(\star)$]}
	\label{lemma:alternate:escape}
	Let a tree $\Tsub$ with models $\{M_u \mid u \in V(G)\}$ represent a connected graph~$G$
		and satisfy condition~\ref{it:compactRep:emptyLeaves}.
	For any vertices \(u,v\) of $G$, \( u \) and \( v \) satisfy the condition~\ref{it:compactRep:escapes} if and only if \( u \) and \(v \) satisfy condition
  \begin{enumerate}
  	\item[$(C3')$] if \( M_u \cap M_v \neq \emptyset \), then \( u \) strongly \escapes \( v \).
  \end{enumerate}
\end{observation}
\begin{proof}
	\forward
	If \( M_u \cap M_v = \emptyset \), then there is nothing to show.
	Otherwise there is a node \( z \) with \( u,v \in \V{z} \).
	Suppose \( u \) \escapes \( v \) at an edge \( (x,y) \in E(\Tsub) \).
	Let \( x' \) be the node on \( \Tsub[z,x] \) with \( v \in \V{x'} \) and maximal distance to \( z \).
	If \( x=x' \), then \( u \) strongly \escapes \( v \) at edge \( (x,y) \).
	Otherwise, since \( \Tsub[z,x] \subseteq M_u \), also \( u \in \V{x'} \).
	Let \( y' \) be the neighbor of \( x' \) further away from \( x' \), which exists because \( \V{x'} \neq \emptyset \) is a non-leaf as implied by~\ref{it:compactRep:emptyLeaves}.
	Thus, \( u \) strongly \escapes \( v \) at the edge \( (x',y') \).

	\backward
	If \( M_u \cap M_v \neq \emptyset \), we have that \( u,v \) have edge \( (x,y) \in E(\Tsub) \) where \( u,v \in \V{x} \) and \( v \notin \V{y} \).
	Thus, in particular, \( u \) \escapes \( v \) at an edge \( (x,y) \).
	If \( M_u \cap M_v = \emptyset \), there is a node \( x \in M_u \setminus M_v \), since the models are non-empty.
	Now, since~\ref{it:compactRep:emptyLeaves} holds, \( u \) is not a leaf in~\( T  \) and thus has at least one neighbor \( y \) where \( v \notin M_y \).
	Thus \( u \) \escapes \( v \) at an edge \( (x,y) \).
\end{proof}

Note that, the non-leaves of a \compact representation are in one-to-one correspondence with the maximal cliques $\CC(G)$.
Namely, we identify the non-leaves with the maximal cliques,
	which implicitly defines the models.
Thus, we often omit the explicit statement of the models.

\begin{observation}
    Let $G$ be a connected graph.
    For any \compact representation $\Tsub$ of $G$,
	\begin{enumerate}
		\item
		for every distinct non-leaves \( x,y \in V(\Tsub) \) there is a vertex \( u \in \V{x} \setminus \V{y} \), and
		\item
		for every edge \( \{x,y\} \in E(\Tsub) \) of non-leaves \( x,y \), there is a vertex \( u \in \V{x} \cap \V{y} \).
	\end{enumerate}
\end{observation}

We (constructively) show that properness and \compactness are essentially equivalent.
To obtain \compactness from properness, we carefully contract edges where a node was used solely to assure properness.
This can involve contracting edges of \( T \) when the vertex sets of the nodes of an edge are comparable, e.g., if they are the same maximal clique.
To obtain properness from \compactness, we subdivide the tree and appropriately extend the models.

\begin{theorem}\sv{[$(\star)$]}
\label{lemma:compact:representation}
	For any connected graph $G$ and tree $T\neq K_1$,
	the graph $G$ is a proper $T$-graph if and only if
	there is re-subdivision $\Tsub$ of $T$
		that is a \compact representation~of~$G$.
\end{theorem}
\begin{proof}
	\forward
	Given a proper representation we construct a \compact representation.
	Let $\{ M_u \}_{u \in V(G)}$ be the models of representation $T$.
	We observe that Condition~\ref{it:compactRep:emptyLeaves}, after a tiny modification of the given representation of~$G$, and Condition~\ref{it:compactRep:escapes} are already satisfied.
	Then we present a modification of the representation whose exhaustive application assures Condition~\ref{it:compactRep:uniqueMaximal} while maintaining conditions~\ref{it:compactRep:emptyLeaves}~and~\ref{it:compactRep:escapes}.

	\noindent\textbf{Maintaining~\ref{it:compactRep:emptyLeaves}.}
	First, for every leaf \( \ell \in V(\Tsub) \), introduce a new leaf \( \ell' \) adjacent to \( \ell \) with \( \V{\ell'} = \emptyset \).
	This clearly preserves properness of the original representation.
	It follows that the resulting tree satisfies Condition~\ref{it:compactRep:emptyLeaves} and is still a re-subdivision of \( T \).
	Here we need that $T\neq K_1$.

	\noindent\textbf{Maintaining~\ref{it:compactRep:escapes}.}
	We observe that properness of models \( \{ M_v \}_{v \in V(G)} \) implies Condition~\ref{it:compactRep:escapes}.
	Consider vertices \( u,v \in V(G) \) with \( M_u \cap M_v \neq \emptyset \).
	Then there are nodes \( x' \in M_u \cap M_v \) and \( y' \in M_u \setminus M_v \)
		(Here we require that the models $M_u$ and $M_v$ are non-empty; which is the case since if either set is empty, the graph $G$ were not connected).
  In particular, there is an edge \((x,y)\) on the path \( \Tsub[x',y'] \) such that \( u \) \escapes \( v \), namely that \( x \in M_u \cap M_v \) and \( y \notin M_v \).
	Thus Condition~\ref{it:compactRep:escapes} holds for every \( u,v \in V(G)\).

	\noindent\textbf{Maintaining~\ref{it:compactRep:uniqueMaximal}.}
	We use the following modification to shrink the tree \( \Tsub \), applicable to any non-leaf node \( z \in V(\Tsub) \) where \( \V{z} \) is not a unique maximal clique.
	Since \( \V{z} \) is not a unique maximal clique, there is a closest non-leaf node \( z' \) such that \( \V{z} \subseteq \V{z'} \).
	Let \( \bar{z} \) be the neighbor of \( z' \) on path \( \Tsub[z,z') \), possibly $\bar{z}=z$.
	Then \( \V{\bar{z}} \subsetneq \V{z'} \).
	Contract the edge \( (\bar{z},z') \) in \( \Tsub \) to a new node \( z^\star \), and if it exists in \( T \), also there.
	Update every model \( M_v \) by removing \( \{\bar{z}, z'\} \) and adding \( \{z^\star\} \) if \( M_v \cap \{\bar{z}, z'\} \neq \emptyset \).
	Since \( \V{\bar{z}} \subseteq \V{z'} \), the new models still represent the adjacency of \( G \).
	Because \( \bar{z} \) and \( z' \) are non-leaves, this modification does not interfere with Condition~\ref{it:compactRep:emptyLeaves}.

	We claim that this modification also preserves Condition~\ref{it:compactRep:escapes} for every \( u,v \in V(G) \).
	By \autoref{lemma:alternate:escape} we may assume that before the modification \( M_u \cap M_v \neq \emptyset \) and that there is an edge $(x,y) \in E(\Tsub)$ such that \( u,v \in \V{x} \) and \( v \notin \V{y} \).
	We show that after the contraction there is an edge \( (x',y') \) where \( u \) \escapes \( v \), which is that \( u \in \V{x'} \) and \( v \notin \V{y'} \).

	If \( y \notin \{\bar{z},z'\} \), then the possible contracted version of \( x \) has \( u \in \V{x} \) and \( y \) is unchanged.
	Thus still \( v \notin \V{y} \) and \( u \) \escapes \( v \).

	If \( x \notin \{\bar{z},z'\} \), then still \( u,v \in \V{x} \).
		Assume that after the contraction \( v \in \V{y} \).
		Because \( \V{\bar{z}} \subseteq \V{z'} \), we have that before the contraction \( v \in \V{z'} \setminus \V{\bar{z}} \) and thus \( y = \bar{z} \).
		Then before the contraction \( v \in \V{x} \cup \V{z'} \setminus \V{\bar{z}} \).
		Observe that \( x,\bar{z},z' \) are on a path in \( \Tsub \).
		A contradiction to that \( M_x \) induces a connected subgraph in tree \( \Tsub \).
		Thus \( y \in \V{y} \), and \( u \) still \escapes~\( v \).

	It remains to consider the case when \( \{x,y\} = \{\bar{z},z'\} \).
	Here, the new vertex \( z^\star \) has \( u \in \V{z^\star} \).
	Since \( y \in \{\bar{z},z'\} \) was not a leaf before, \( y \) has a non-contracted neighbor \( y' \notin \{x,y\} \) now adjacent to \( z^\star \).
	Because \( \Tsub \) is a tree, \( y' \notin M_v \) before and after the modification.
	Thus \( u \) \escapes \( v \) an edge \( (z^\star,y') \).

	If a non-leaf node \( z \in V(\Tsub) \) violates Condition~\ref{it:compactRep:uniqueMaximal}, that \( \V{z} \) is not a unique maximal clique, then the above modification applicable to \( z \).
	After the exhaustive application of this modification, the resulting representation satisfies all three conditions.
	Since the above modification shrinks the size of \( \Tsub \) by one, this procedure terminates.

	\medskip

	\backward
	Given a compact representation, we construct a proper representation.
	Let us first make some easy modification to the given $\Tsub$-representation such that our main step becomes more convenient.

	First, for each pair of distinct vertices \(u,v\) with \( M_u = M_v \), we temporarily remove~$v$ and $M_v$ from our graph and representation (we will reintroduce these later). Note that, since \(M_u = M_v\), it must be that \(u\) and \(v\) are twins, i.e., they have the same closed neighborhoods (after removal of both of them). Let $Z$ the set of vertices we have removed in this step and let~$\tilde{G} = G - Z$.

	Second, we subdivide every edge of \( \Tsub \) once and adjust the models of vertices accordingly, which will be of use later. %
	Condition~\ref{it:compactRep:escapes} remains satisfied, since if \( u \) \escapes \( v \) at an edge \( (x,y) \) which we subdivide into a path \( x,z,y \), then \( u \) \escapes \( v \) at the edge~\( (x,z) \).

	Now, we modify the representation to a proper one.
	For every pair of vertices \( {(u,v) \in \binom{V(\tilde{G})}{2}} \) with \( M_u \subseteq M_v \) we need to modify the representation to make \( u \) and \( v \) proper.
	We will do so by extending the model \( M_u \) beyond the model \( M_v \) at some edge \( (x,y) \in T \), where \( u,v \in \V{x} \setminus \V{y} \).
	There may be other pairs of vertices \( (u',v') \in \binom{V(\tilde{G})}{2} \) with \( M_{u'} \subseteq M_{v'} \) where we want to extend the model \( M_{u'} \) beyond \( M_{v'} \) at the same edge \( (x,y) \).
	To this end, for every edge \( (x,y) \), let us collect these `requests to extend' at \( (x,y) \) by an initially empty directed graph \( G_{x,y} \) over the vertex set \( V(\tilde{G}) \).
	An edge \( (u,v) \) in the graph \( G_{x,y} \) represents the request to extend \( M_u \) further than \( M_v \).
	We add edges \( (u,v) \in \binom{(\V{x} \setminus \V{y})}{2} \) to \( E(G_{x,y}) \) if \( M_u \subseteq M_v \).
	We assure in the following that \( G_{x,y} \) has no directed cycle, which is true initially.

  \begin{claim*}
		The graph \( G_{x,y} \) has no directed cycle.
	\end{claim*}
	\begin{proof}
		Consider a pair of vertices \( (u,v) \) with \( M_u \subseteq M_v \).
		Then, according to Condition~\ref{it:compactRep:escapes}, there is an edge \( (x,y) \) such that \( u,v \in \V{x} \) and \( v \notin \V{y} \).
		Since \( M_u \subseteq M_v \), also \( u \notin \V{y} \).
		Thus \( (u,v) \) meets the preconditions and we can add it to \( G_{x,y} \).
		Assume, for the sake of contradiction, that the edge \( (u,v) \) introduces a directed cycle \( u,v,w_1,\dots,w_s,u \) in \( G_{x,y} \).
		The preconditions to add each edge of the cycle imply for the models \( M_u \subseteq M_v \subseteq M_{w_1} \subseteq \dots \subseteq M_{w_s} \subseteq M_u \) and thus \( M_u = M_v \).
		However, this contradicts our modification to avoid equal models \( M_u = M_v \) at the very beginning.
		Thus, adding \( (u,v) \) does not introduce a directed cycle.
	\renewcommand{\qed}{\hfill$\lhd$}\end{proof}

	What remains is to implement the ordering of models \( M_v \) at an edge \( (x,y) \) for \( {v \in \V{x} \setminus \V{y}} \) contained in \( G_{x,y} \).
	Recall that we initially subdivided every edge, which means that there are no \( v' \in \V{y} \setminus \V{x} \) to be considered.
	Let \( C_{x,y} = u_s,\dots,u_0 \) be a topological ordering of \( G_{x,y} \).
	Such an ordering exists, since \( G_{x,y} \) is a directed acyclic graph.
	We subdivide the edge \( (x,y) \) into a path \( x, z_1, \dots, z_s, y \).
	For \( i \in [s] \) we extend the model \( M_{u_i} \) by adding nodes \( \{z_1,\dots,z_i\} \) to it.
	For every other vertex \( w \in V(G) \) with model \( M_{w} \) containing both \( x \) and \( y \), we add the whole set of new nodes \( \{z_1,\dots,z_s\} \) to \( M_w \).
	Then, after the modification, every model \( M_v \) for \( v \in V(G) \) is still a connected subgraph of the subdivided tree \( \Tsub \).
	Moreover, every already proper pair of vertices \( u,v \) (that is, where there is a node \( z \in M_u \setminus M_v \) and a node \( z' \in M_v \setminus M_u \)) remains proper, since the above modification does not change \( \V{z''} \) for any node \( z'' \in V(\Tsub) \).

	We apply this modification to every edge \( (x,y) \) of \( T \).
	Let \( T^\star \) be the resulting tree and \( M_v^\star \) the resulting model of a vertex \( v \in V(G) \).
	It remains to show that if for \( (u,v) \in V(G) \) where initially \( M_u \subseteq M_v \) there is now a node \( y \in M_u^\star \setminus M_v^\star \).
	If \( M_u \subseteq M_v \), then as seen above there is and edge \( (x,y) \) and we add edge \( (u,v) \) to \( G_{x,y} \).
	Thus, in \( G_{x,y} \) we have that \( u \) is ordered after \( v \).
	Thus, there is a node \( z_i \) from the subdivision of \( (x,y) \) with \( z_i \in M_u^\star \setminus M_v^\star \).

  Finally, we reintroduce the vertices of \(Z\) and appropriate models for them.
  Note that for each \(v \in Z\), there is a \(u \in V(G) \setminus Z\) such that \(u\) and \(v\) are twins.
  We now iteratively consider each \(v \in Z\), let \(u \in V(G) \setminus Z\) be a twin of \(v\), remove \(v\) from \(Z\), and define the new \(M_v\) as a copy of \(M_u\).
  Note that besides the relation of \(M_u\) and \(M_v\), our representation is proper.
	We now modify \(M_u\) and \(M_v\) so that we again have a proper representation.
	By Condition~\ref{it:compactRep:emptyLeaves}, no node of either \( M_u \) or \( M_v \) is a leaf of \(\Tsub\).
	Thus there are two edges \( (x,y) \) and \( (x',y') \) in the tree \( \Tsub \) where \( u \) \escapes \( v \), meaning that \( u,v \in \V{x} \cap \V{x'} \) but \( u,v \notin \V{y} \cup \V{y'} \).
	We subdivide the edge \( (x,y) \) and \( (x',y') \) to a path \( x,z,y \) respectively to a path \( x',z',y' \), and adjust the models accordingly (i.e., if both $x,y \in M_w$ for a vertex~$w$, we add $z$ to $M_w$; analogously for $x',y',z'$). %
	Additionally, extend the model \( M_u \) by \( z \) and the model \( M_v \) by \( z' \).
	Then \( M_u \nsubseteq M_v \) and \( M_v \nsubseteq M_u \).
	Further, the adjacency of \( G \setminus Z\) remains satisfied and the properness of our representation is restored.

	After emptying the set \(Z\) in this way we end up with a proper \( \Tsub \)-representation of \(G\) where \( \Tsub \) is a subdivision of \(T\).
\end{proof}

Thus, instead of finding a proper representation, we search for a \compact representation.
The actual `properness' is hidden in the condition~\ref{it:compactRep:escapes}, and we may refer to this condition as \define{properness}.
See also examples in \autoref{figure:examples}.

Our algorithm further relies on the following property of the models $M_\Gamma$ of the (connected) components of $G-\V{\y}$ for some non-leaf $\y$; see also \autoref{figure:partition}(a).
Let $\KK{\y}$ (w.r.t.\ graph $G$) be the vertex sets of the components of $G - \V{\y}$.
We note that $N(\Gamma) \subseteq \V{\y}$ for every $\Gamma \in \KK{\y}$.
Let a node $\y$ be an \define{\eye} if it is a neighbor of a leaf or if it is a branching node.

\newcommand{\K}{\Gamma}
\begin{lemma}\sv{[$(\star)$]}
	\label{lemma:K:of:m:partition}
    Let $G$ be a connected graph.
	For any \compact representation \( \Tsub \) of \( G \) and any non-leaf node \( \y \in V(\Tsub) \),
	\begin{enumerate}
		\item\label{lemma:K:of:m:partition:1}
		\( \{\y\} \) and \( M_\K \) for \( \K \in \KK{\y} \) partition the non-leaves of \( \Tsub \), and
		\item\label{lemma:K:of:m:partition:3}
		each partition \( M_\K \) contains an \eye, hence $|\KK{\y}|\leq |V(T)|$.
	\end{enumerate}
\end{lemma}
\begin{proof}
	\ref{lemma:K:of:m:partition:1}.
	We show that \( M_\K \) for \( \K \in \KK{\y} \) and \( \{\y\} \) are disjoint sets, and cover every non-leaf node of \( V(\Tsub) \).
	Assume that two components \( \K, \K' \in \KK{\y} \) intersect in a node \( \y \in V(\Tsub) \).
	Then there are (possibly equal) vertices \( u \in \V{\y} \cap \K \) and \( v \in \V{\y} \cap \K' \).
	Thus, \( \K \) neighbors \( \K' \), and hence \( \K = \K' \).
	Consider a non-leaf node \( x \in V(\Tsub) \).
	Then there is a vertex \( u \in \V{x} \setminus \V{\y} \).
	Thus there is a component containing of \( \KK{\y} \) containing \( u \).
	If there is a component \( \K \in \KK{\y} \) where \( \y \in M_\K \), then there is a vertex \( u \in \V{\y} \cap \K \).
	This implies the contradiction \( u \in N(\K) \cap \K \).

	\ref{lemma:K:of:m:partition:3}.
	Consider a component $\K \in \KK{\y}$.
	Assume, for the sake of contradiction, that the model \( M_\K \) neither neighbors a leaf nor contains a branching node.
	Then there is a non-leaf neighbor \( x \) of \( M_\K \) such that $x$ and $\y$ are separated by $M_k$.
	Further, there is a component \( \Gamma_{x} \in \KK{\y} \setminus \{\K\} \) containing non-empty \( \V{x} \).
	Moreover, there is a vertex \( v \in N(\Gamma_x) \cap \V{\y} \),
		which has a model $M_v$ containing nodes \( x \) and \( \y \).
	Since \( M_\K \) contains only subdivision nodes, model \( M_v \) has to contain the whole \( M_{\K} \).
	Then \( \K \) contains at least one vertex \( u \) with model \( M_u \subseteq M_v \).
	Thus \( u \) does not \escape \( v \) in contradiction to a \compact representation.

	Now the final remark that $|\KK{\y}|\leq |V(T)|$ follows easily.
	Each \eye of $T$ is contained in at most one connected component and each component contains a maximal clique whose node is an \eye.
	Thus the number of components $|\KK{\y}|$ is bounded by the number of \eyes in $T$ which is $|V(T)|$.
\end{proof}

%% file: proper-tree-graphs.tex
\newcommand{\Y}{Y}

\section{Finding a Compact Representation}
\label{sec:finding-compact-reps}

In this section, we prove Theorem~\ref{lemma:algorithm}; namely, we establish our \cFPT algorithm.
Throughout the discussion, we assume $G$ is connected, and handle disconnected graphs within the final proof.
From Section~\ref{sec:compact-reps}, it suffices to check for a \compact representation $\Tsub$.
In Subsection~\ref{subsection:surround:nodes}, we establish the concept of \emph{surrounded} nodes, which leads, in Subsection~\ref{subsection:chains}, to the \emph{chains} that necessarily form paths in any \compact tree representation.
We establish that the chains (composed of surrounded nodes), and the remaining non-surrounded nodes are only quadratically many in the size of the desired tree $T$.
In Subsection~\ref{sec:template}, we formalize the way these pieces fit together as \emph{templates}.
Finally, Subsection~\ref{sec:norm} contains the algorithm establishing Theorem~\ref{lemma:algorithm}. It proceeds by enumerating candidate templates and (non-trivially) testing whether a template admits a \compact representation via a bottom-up procedure.

\newcommand{\Org}{{S}}

\newcommand{\lc}{\lambda}
\newcommand{\yc}{\gamma}
\newcommand{\rc}{\rho}
\subsection{Surrounded Nodes}
\label{subsection:surround:nodes}

We establish conditions for arbitrary nodes $\ell,\y,r$ that determines the relative position of $\ell,\y,r$ in any representation $\Tsub$,
	a relation which we denote as $(\ell,\y,r)$ surrounding.
Clearly, this positioning is unlikely to be possible for every triple $(\ell,\y,r)$ since this would yield a polynomial-time algorithm.
However, by carefully crafting our first two requirements, we may still relatively position almost all nodes $\ell,\y,r$.
We only fail for a few nodes $\y$, at most quadratic in the size of the host tree $T$, hence our parameter.

\input{figure-surround}

\begin{definition}
	Consider non-leaves \( \ell,\y,r \) of \( \Tsub \).
	There is a component $\K_{\ell}\in\KK{\y}$ containing $\V{\ell}\setminus\V{\y}$,
		likewise a component $\K_{r}\in\KK{\y}$ containing $\V{r}\setminus\V{\y}$.
	Then $(\ell,\y,r)$ is a \define{surrounding triple}, if the following conditions are met:
	\begin{enumerate}
		\item[\labeltext{(1)}{it:surround:2}]
			If $\{\K_\ell,\K_r\} = \KK{\y}$, then
			$\V{\y} = N(\K_\ell) \cup N(\K_r)$ or $N(\K_\ell) \cap N(\K_r) = \emptyset$;
		\item[(2)] if \( \{\K_\ell,\K_r\} \subsetneq \KK{\y} \),
		\begin{enumerate}
			\item[\labeltext{(2A)}{it:surround:V:y}]
				$\V{\y} = N(\K_\ell) \cup N(\K_r)$, and
			\item[\labeltext{(2B)}{it:surround:subset}]
				for every $\K \in \KK{\y} \setminus \{\K_\ell,\K_r\}$ we have: $N(\K) \subseteq N(\K_\ell) \cap N(\K_r)$; and
		\end{enumerate}
		\item[\labeltext{(3)}{it:surround:minimality}]
		for every \( \ell',r' \) that satisfy (1), (2A), and (2B) where \( \K_{\ell} = \K_{\ell'} \) and  \( \K_r = \K_{r'} \),
			we have \( \V{\ell'} \cap \V{\y} \subseteq \V{\ell} \cap \V{\y} \) and \( \V{r'} \cap \V{\y} \subseteq \V{r} \cap \V{\y} \).
	\end{enumerate}
\end{definition}

Note, that the definition does \stress{not} depend on the considered representation $\Tsub$.
Importantly, $\K_\ell\neq\K_r$ is (implicitly) required by condition~\ref{it:surround:2}.
We say that \define{$\y$ is surrounded}, if a triple $(\ell,\y,r)$ is a surrounding triple for some nodes $\ell,r$.
See \autoref{figure:examples} for examples.

For each node $\y$, the connected components $\K_\ell$ and $\K_r$ satisfy or falsify the first two conditions independently of the precise maximal cliques $\V{\ell}$ and $V_r$.
However, condition~\ref{it:surround:minimality} requires $\V{\ell}$ and $\V{r}$ to be the closest ones to $\V{\y}$.
In many cases condition~\ref{it:surround:minimality} implies that $\ell$ and $r$ directly neighbor $\y$.
In fact, for a surrounded node $\y$, there are sets of nodes $L$ and $R$ that exactly localize the nodes $\ell$ and $r$ forming a surrounding triple with~$\y$.
Formally, $L,R$ are \define{$\y$-\guards}:
A set of non-leaves \( L \subseteq V(\Tsub) \) is a \( \y \)-\define{\guard} if \( L \cup \{\y\} \) is connected, and \( \y \) is adjacent to a node $\ell \in L$ such that \( \{\ell\}=L \) or \( \ell \) is a branching node of \( \Tsub \); see \autoref{figure:L:R}(b).

\begin{lemma}\sv{[$(\star)$]}\label{lemma:surround:characterization}
	Let a tree $\Tsub$ be a \compact representation of a connected graph \( G \).
	Let \( \y \) be surrounded.
	There are distinct \( \y \)-guards \( L \) and \( R \) such that
		$(\ell,\y,r)$ is surrounding if and only if \( (\ell,r) \in (L \times R) \cup (R \times L) \).
	Moreover there is an $\Oh(t^3 n^3)$-time algorithm that determines sets $L,R$ for every surrounded node $\y$; where $t=|V(T)|$ and $n=|V(G)|$.
\end{lemma}
\begin{proof}
	\newcommand{\characterizationA}{(A)\xspace} %
	\newcommand{\characterizationB}{(B)\xspace} %
	\newcommand{\characterizationI}{(C)\xspace} %
	\newcommand{\characterizationD}{(D)\xspace} %
	\newcommand{\characterizationE}{(E)\xspace} %
	Let us first show (I) the characterization by $\y$-guards $L$ and $R$,
	then (II) the algorithm to determine $L$ and $R$.

	\medskip

	(I)
	We characterize \( \y \)-guards \( L \) and \( R \) that contain nodes that surround a fixed \( \y \) by the following observations \characterizationA, \characterizationB, \characterizationI, \characterizationD, \characterizationE that constitute the forward direction~\forward.
	As a final step we show the backward direction \backward.

	\characterizationA We claim that every \( \ell,r \) that surround \( \y \) are from different subtrees.
	Assuming the contrary, let \( \ell, r \) be contained in the same subtree of \( \y \),
		rooted at neighbor \( \y' \) of \( \y \).
	Then there is a vertex \( u \in \V{\y} \setminus \V{\y'} \), which implies that \( \V{\y} \nsubseteq \V{\ell} \cup \V{r} \), or equally \( \V{\y} \nsubseteq N(\K_\ell) \cup N(\K_r) \).
	Then \( \KK{\y} = \{\K_\ell,\K_r\} \) as otherwise condition~\ref{it:surround:V:y} fails.
	Because of condition~\ref{it:surround:2} then \( N(\K_\ell) \cap N(\K_r) = \emptyset \).
	Recall that \( M_\K \) for \( \K \in \KK{\y} \) partitions the non-leaves of \( V(\Tsub) \setminus \{\y\} \) as seen in \autoref{lemma:K:of:m:partition}.
	Thus the non-leaf \( \y' \) belongs to either \( M_{\K_\ell} \) or \( M_{\K_r} \). %
	Also \( \y' \) does not belong to both sets \( M_{\K_\ell}, M_{\K_r} \), since otherwise we have the contradiction \( \emptyset \neq \V{\y'} \cap \V{\y} \subseteq N(\K_\ell) \cap N(\K_r) \).
	By symmetry, we may assume that \( \y' \in M_{\K_\ell} \) and \( \y' \notin M_{\K_r} \).
	Let \( z \) be the neighbor of \( \y' \) that is on the path from \( \y' \) to \( M_{\K_r} \).
	Then there is at least one vertex \( v \in \V{z} \cap \V{\y'} \) where \( v \in N(\K_r) \).
	However, since \( v \in \V{\y'} \), also \( v \in N(\K_\ell) \), a contradiction to \( N(\K_\ell) \cap N(\K_r) = \emptyset \).

	\characterizationB
	We claim that there are only two subtrees \( T_L, T_R \) of \( \y \) such that for every surrounding triple $(\ell,\y,r)$ nodes $\ell$ and $r$ are in the subtrees \( T_L \) and \( T_R \).
	Assuming the contrary, there are surrounding triples $(\ell,\y,r)$ and $(\ell',\y,r')$ where $\ell$ and $\ell'$ are from distinct subtrees of $\y$, though possibly \( r = r' \).
	In particular \( \K_\ell, K_{\ell'} \), and \( \K_r \) are distinct components of \( \KK{\y} \).
	By applying condition~\ref{it:surround:subset} for triples \( (\ell',\y,r') \) and \( (\ell,\y,r) \), we have that \( N(\K_\ell) \subseteq N(K_{\ell'}) \subseteq N(\K_r) \)
		and analogously \( N(\K_r) \subseteq N(K_{\ell'}) \subseteq N(\K_\ell) \), which implies that \( N(\K_\ell) = N(\K_r) \).
	There is a node \( u \in \V{\y} \setminus \V{\ell} \), and since \( N(\K_\ell) = N(\K_r) \) also \( u \notin \V{r} \).
	Then, however, \( u \) contradicts that \( \ell,r \) surround \( \y \) as it falsifies condition~\ref{it:surround:V:y}.

	\characterizationI
	We claim that for every nodes $\ell,r$ such that $(\ell,\y,r)$ is a surrounding triple have models \( M_{\K_\ell} \) and \( M_{\K_r} \) that neighbor \( \y \).
	Models \( M_{\K_\ell} \) and \( M_{\K_r} \) are from distinct subtrees of \( \y \) as implied by \characterizationA.
	Recall that the models of components \( \KK{\y} \) partition the non-leaves of \( \Tsub - \y \) as seen in \autoref{lemma:K:of:m:partition}.
	Assume, for the sake of contradiction, that say \( M_{\K_\ell} \) does not neighbor \( \y \).
	Then there is a component \( \K \in \KK{\y} \setminus \{\K_\ell,\K_r\} \) with model \( M_\K \) separating \( M_{K_{\ell}} \) and \( \y \).
	Consider a vertex \( u \in \V{\y} \setminus \V{r} \).
	To satisfy condition~\ref{it:surround:V:y} we have \( u \in \V{\ell} \).
	Then \( u \in N(\K) \) and hence \( u \in N(\K_\ell) \cap N(\K_r) \) by condition~\ref{it:surround:subset},
		which implies \( u \in \V{r} \). Contradiction.

	\characterizationD
	For a subtree $T_L$ of node $\y$ consider \( L \), the set of nodes $\ell$ which have some \( r \) such that \( \ell,r \) surround \( \y \).
	We claim that \( L \) is connected and neighbors \( y \).
	Consider nodes \( \ell \in L \) and \( \ell' \notin L \) where $(\ell,\y,r)$ is surrounding and \( (\ell,\ell',\y) \) is \( \Tsub \)-ordered.
	Recall from \characterizationI that \( \K_\ell \) neighbors \( \y \) and hence \( K_{\ell'} = \K_\ell \).
	It suffices to show that also \( (\ell',\y,r) \) is surrounding.
	Conditions \ref{it:surround:2}, \ref{it:surround:V:y} and \ref{it:surround:subset} are still true since they only depend on the components \( \K_{\ell'} \) and \( \K_r \).
	Assuming condition~\ref{it:surround:minimality} fails for \( (\ell',\y,r) \) because of nodes \( \ell'',r'' \),
		implies that also \( \ell'',r'' \) falsify condition~\ref{it:surround:minimality} for \( (\ell,\y,r) \), since \( \V{\ell} \cap \V{\y} \subseteq \V{\ell'} \cap \V{\y} \). %

	\characterizationE
	To show that \( L \) is a \( \y \)-guard, it remains to show that if a neighbor \( \ell \in L \) of \( \y \) is a subdivision node, then \( L=\{\ell\} \).
	Assuming the contrary, there is a path \( \ell^\star \ell \y \) in \( \Tsub \) such that \( \ell^\star,\ell \in L \) as implied by \characterizationD.
	Because \( N(K_{\ell^\star}) \neq \emptyset \), there is a node \( v \in \V{\ell^\star} \cup \V{\y} \).
	Further there is a node \( u \in \V{\ell} \setminus \V{\ell^\star} \), which in order to escape \( v \) has \( u \in \V{\y} \).
	Then, however, \( \V{\ell^\star} \cap \V{\ell} \subsetneq \V{\ell} \cap \V{\y} \) in contradiction to condition~\ref{it:surround:minimality}.

	Finally, let \( R \), analogously to \( L \), be the nodes \( r \) of the other subtree \( T_R \) such that \( r \) and some \( \ell \) surround \( \y \).
	By symmetry to \( L \), also \( R \) is a \( y \)-guard.
	According to \characterizationA and \characterizationB, every \( \ell,r \) that surround \( \y \) have \( \ell \in L \) and \( r \in R \) or vice versa \( \ell \in R, r \in L \).

	\backward
It remains to show the reverse, that every triple $(\ell',\y,r)$ for \( \ell' \in L \) and \( r \in R \) is a surrounding triple.
Assuming that $\ell' \in L$ and $r \in R$ there are some nodes $\ell \in L$ and $r' \in R$
	such that $(\ell,\y,r)$ and $(\ell',\y,r')$ are surrounding triples.
Because $\ell,\ell'$ are in the same subtree of $\y$ as seen in \characterizationI,
	we have that $\K_{\ell} = \K_{\ell'}$,
	likewise it follows that \( \K_r = \K_{r'} \).
Thus conditions~\ref{it:surround:2}, \ref{it:surround:V:y} and \ref{it:surround:subset} for $(\ell',\y,r)$ are directly satisfied.
Regarding condition~\ref{it:surround:minimality}, assume nodes $\ell'',r''$ that surround $\y$.
We may assume that \( \K_{\ell''} = \K_\ell \) and \( \K_{r''} = \K_r \).
Then \( \V{\ell''} \cap \V{\y} \subseteq \V{\ell'} \cap \V{\y} \) is true, since \( \ell',r' \) surround \( \y \).
Likewise \( \V{r''} \cap \V{\y} \subseteq \V{r} \cap \V{\y} \) because \( \ell,r \) surround \( \y \).
Therefore also condition~\ref{it:surround:minimality} is true for $(\ell',\y,r)$.

	\medskip

(II)
For simplicity, let us output empty $L,R$ if a node $\y$ is not surrounded.
We use the following condition which simplifies and is equivalent to condition~\ref{it:surround:minimality}.
\begin{itemize}
\item[$(3')$]
for every \( \ell',r' \) that satisfy (1), (2A), and (2B) where \( \K_{\ell} = \K_{\ell'} \) and  \( \K_r = \K_{r'} \),
	we have $|\V{\ell'} \cap \V{\y}| \leq |\V{\ell} \cap \V{\y}|$ and $|\V{r'} \cap \V{\y}| \leq |\V{r} \cap \V{\y}|$.
\end{itemize}
Condition (3) easily implies $(3')$ since the subset relation $\V{\ell'} \cap \V{\y} \subseteq \V{\ell} \cap \V{\y}$ implicates the cardinality relation $|\V{\ell'} \cap \V{\y}| \leq |\V{\ell} \cap \V{\y}|$.
For the reverse direction, recall that conditions~\ref{it:surround:2}, \ref{it:surround:V:y} and \ref{it:surround:subset} imply that there are two components $\K_L,\K_R \in \KK{\V{\y}}$
	such that $\ell,r$ form a surrounding triple $(\ell,\y,r)$ only if $(\ell,r) \in (L\times R) \cup (R \times L)$
	as seen in observation~\characterizationI.
By observation~\characterizationD its model $M_{\K_\ell}$ has to neighbor $\y$,
	and hence let us consider the neighbor $\ell^\star \in M_{\K_\ell}$ of $\y$.
It suffices to show that for every other node $\ell' \in M_{\K_L}$ the relation $|\V{\ell'} \cap \V{\y}| \leq |\V{\ell^\star} \cap \V{\y}|$ implies $\V{\ell'} \cap \V{\y} \subseteq \V{\ell^\star} \cap \V{\y}$.
Assuming the contrary, then there is a vertex $u$ where $\ell',\y \in M_u$ but $\ell^\star\notin M_u$,
	which contradicts connectivity of model $M_u$.
Therefore condition $(3')$ is equivalent to \ref{it:surround:minimality}.

Now we show how to determine $L,R$ for a fixed node $\y$ in time $\Oh(t^3n^2)$,
	which since there are only $\Oh(n)$ maximal cliques, hence nodes, shows the desired runtime.
First, for every other node $\ell\neq\y$ determine the size of the intersection $|\V{\ell}\cap\V{r}|$ in time $\Oh(n^2)$.
Then determine the components $\Gamma \in \KK{\y}$ together with their neighborhood $N(\K)\subseteq \V{\y}$ by breath-first-search in time $\Oh(n^2)$.
Note that there are at most $t=|V(T)|$ components in $\KK{\y}$ as seen in \autoref{lemma:K:of:m:partition}.
Try all $t^2$ possible pairs of components and check if they satisfy conditions \ref{it:surround:2}, \ref{it:surround:V:y} and \ref{it:surround:subset} in time $\Oh(tn)$ for every pair of components; hence $\Oh(t^3n)$ for this step.
Either we conclude that $\y$ is not surrounded, or we obtain components $\K_L$ and $\K_R$ such that every nodes $\ell,r$ such that $(\ell,\y,r)$ is surrounding have $(\ell,r) \in  (M_{\K_L} \times M_{\K_R}) \cup (M_{\K_R} \times M_{\K_L})$.
It remains to find those $\ell \in \K_\ell$ that maximize the cardinality of their intersection $|\V{\ell}\cap\V{\y}|$ as previously observed; and analogously for $r\in\K_r$.
To find those, we may simply iterate through $\ell \in K_\ell$ and remember those with maximal intersection.
Since we determined the intersections before, this step requires $\Oh(n)$ time.
To summarize, to determine $L,R$ for a fixed node $\y$, we need time $\Oh(t^3n^2)$;
thus for every node $\y$, a total time of $\Oh(t^3n^3)$.
\end{proof}

The \define{guards of \( \y \)} are such distinct \( \y \)-guards \( L \) and \( R \) that precisely characterize its surrounding triples.
It is worth noting that a node $\y$ that is surrounded by subdivision nodes has singleton guards $\{\ell\}$ and $\{r\}$, which then must be neighbors $\y$ in any representation.

\medskip

Surprisingly
there is also a quadratic bound in $|V(T)|$ on the number of \emph{not} surrounded nodes.
To show this, the main difficulty is that the conditions \ref{it:surround:V:y} and \ref{it:surround:subset} fail for a subdivision node $\y$ due to some \define{remote} component $\Gamma$,
	which is a connected component $\Gamma\in\KK{\y}\setminus\{\Gamma_\ell,\Gamma_r\}$.
Here, let $\y \in \Tsub(x,z)$ for some neighbors $x,z\in V(T)$.
To cope with these remote components, we use three ingredients:
\begin{itemize}
\item
A component $\Gamma \in \KK{\y}$ that falsifies the conditions in question relates to $\KK{x}$ or $\KK{z}$:
Its model $M_\Gamma$ must be outside of $\Tsub[x,z]$.
By examining the nodes $\y'$ that have $\Gamma$ as a component, it follows that either $\Gamma\in\KK{x}$ or $\Gamma\in\KK{z}$.
\item
We use \autoref{lemma:K:of:m:partition} on components $\Gamma \in \KK{x}$, likewise for $\KK{z}$:
The models of the components $\Gamma \in \KK{x}$ partition the non-leaves $\Tsub$,
	and each of its models $M_\Gamma$ contains an \eye. %
That means that $\KK{x}$ contains at most $|E(T)|$ components.
\item
A component $\Gamma\in\KK{x}$ can be remote for at most one node $\y$ on the path $\Tsub(x,z)$,
	and hence falsify the surround conditions for at most one $\y$ on that path.
This yields a simple $2|E(\Tsub)|$-bound for not surrounded nodes on that path; see also \autoref{figure:falsify:only:one}(c).
\end{itemize}
By incorporating these ideas in a more careful manner we obtain the following bound:

\input{figure-components}

\newcommand{\remote}{remote\xspace}
\begin{lemma}\sv{[$(\star)$]}
\label{lemma:bounding:not:surrounded}
	Let subdivision \( \Tsub \) of a tree \( T \) be a \compact representation of a connected graph \( G \).
	There are at most \( |E(T)|^2 + 1 \) non-leaves of\/ \( V(\Tsub) \) that are not surrounded.
\end{lemma}
\begin{proof}
	Regarding \eyes of \( V(\Tsub) \), our bound of not-surrounded nodes is simply the number of \eyes in \( \Tsub \), which still at most \( |E(T)|+1 \).
	In turn, we limit the number of not-surrounded non-\eye subdivision nodes as follows.
	For fixed \eyes \( x,\z \) that are connected by a path of subdivision nodes, we upper bound the number of not surrounded subdivision nodes of that path by \( |E(T)|-1 \).
	Thus beside \eyes, there are at most \( |E(T)| \cdot (|E(T)|-1) \) further not surround nodes, which yields the desired bound.

	Let \( \Tsub(x,\z) \) be the path of inner nodes connecting \( x \) and \( \z \).
	Consider the components of \( \KK{\y} \) for a node \( \y \in \Tsub(x,\z) \) with neighbors \( \ell,r \) such that \( (x,\ell,\y,r,\z) \) is \( \Tsub \)-ordered.
	Then the two components \( \K_\ell \) and \( \K_r \) with models containing \( \ell \) respectively \( r \) neighbor~\( \y \).
	Recall that the model of each component of \( \KK{\y} \) contains an \eye, as seen in \autoref{lemma:K:of:m:partition}.
	Thus we know that \( x \in M_{\K_\ell} \) and \( \z \in M_{\K_r} \).
	In turn, other components \( \K \in \KK{\y} \setminus \{\K_\ell,\K_r\} \) do not contain \( \Tsub[x,\z] \).
	We denote such a component \( \K \in \KK{\y}\setminus \{\K_\ell,\K_r\} \) where additionally \( \K \notin \KK{\ell} \) or \( \K \notin \KK{r} \) as \( \y \)-\define{\remote}.

	Our goal is to upper bound the number of not-surrounded subdivision nodes $\y\in \Tsub(x,z)$ by the number of components of $\KK{x}$ and $\KK{z}$.
	The first step is to observe that a not-surrounded subdivision nodes $\y$ implies a $\y$-remote component, and then observe that these \remote components are distinct for different $\y,\y' \in \Tsub(x,z)$ and bounded by the number of \eyes in $\Tsub$.

	\begin{claim}\label{clm:not-surroundedImliesRemoteComponent}
		Let \( \y \) be a not-surrounded subdivision node of \( \Tsub(x,\z) \), then there exists a \( \y \)-\remote component.
	\end{claim}
  	\begin{proof} %
		Consider a not-surrounded node \( \y \) of \( \Tsub(x,\z) \).
		Let \( \ell \) and \( r \) be the neighbors of \( \y \) such that \( (x,\ell,\y,r,\z) \) is \( \Tsub \)-ordered.
		In particular $(\ell,\y,r)$ is not surrounding as implied by \autoref{lemma:surround:characterization}.
		Surround condition~\ref{it:surround:minimality} cannot fail alone.
		Thus one of the conditions \ref{it:surround:2}, \ref{it:surround:V:y}, or \ref{it:surround:subset} fails.
		\begin{itemize}
			\item
			Assuming condition~\ref{it:surround:2} fails, then \( \V{\y} \nsubseteq N(\K_\ell) \cup N(\K_r) \) and \( N(\K_\ell) \cap N(\K_r) \neq \emptyset \).
			The first implies that there exists a vertex \( u \) with model \( M_u = \{\y\} \).
			The latter implies that there exists a vertex \( v \in N(\K_\ell) \cap N(\K_r) \) and hence with model \( M_v \) containing \( \{\ell,\y,r\} \).
			We reach a contradiction, since now \( u \) does not \escape \( v \).
			Thus, condition~\ref{it:surround:2} cannot fail.
			\item
			Assuming condition~\ref{it:surround:subset} fails, there is a component \( \K \in \KK{\y} \setminus \{\K_\ell,\K_r\} \) such that \( N(\K) \notin N(\K_\ell) \cap N(\K_r) \).
			However since \( M_\K \) does not neighbor \( \y \), we have that \( N(\K) \in \V{\ell} \) or \( N(\K) \in \V{r} \).
			This implies that \( \K \) is \( \y \)-\remote.
			\item
			Finally, assuming condition~\ref{it:surround:V:y} fails, we have
				\[ \V{y} \nsubseteq N(K_\ell) \cup N(\K_r) \text{, and } \K \in \KK{y} \setminus \{\K_\ell,\K_r\} . \]
			The first, implies that there exists a vertex \( u \in \V{y} \setminus (\V{\ell} \cap \V{r}) \).
			If \( \K \in \KK{\ell} \cap \KK{r} \), then there is a vertex \( v \in N(\K) \) with \( v \in \V{\ell} \cap \V{\y} \cap \V{r} \) such that \( u \) does not \escape \( v \).
			However \( \K \in \KK{\ell} \cup \KK{r} \), since \( \K \notin \{\K_\ell,\K_r\} \), which makes \( \K \) a \( \y \)-\remote component.
		\end{itemize}
		Therefore in every case, where \( \y \) is not surrounded a \( \y \)-\remote component exists.
	\renewcommand{\qed}{\hfill$\lhd$}\end{proof}

	\begin{claim}\label{clm:M_K:and:S_K:connected}
		Consider a connected $\K\subseteq V(G)$.
		We claim that $M_\K \cup \Org_\K \) for \( \Org_\K \coloneqq \{ \y \mid \K \in \KK{\y} \}$ is connected.
	\end{claim}
	\begin{proof} %
		As seen in \autoref{lemma:M:K:is:connected}, model $M_\K$ is connected.
		Assuming, for the sake of contradiction, that $M_\K\cup\Org_K$ is not connected then implies that there are $\Tsub$-ordered nodes $(\y_1,\y_2,\y_3)$ where $\y_1\in M_\K\cup\Org_k$ while $\y_2\notin M_k\cup\Org_K$ and $\y_3\in\Org_K$.
		Because $M_\K\neq\emptyset$, we may assume that $\y_1\in M_\K$.
		Further we may assume that $\y_1$ is the nearest node of $\y_3$ that is part of $M_\K$.

		Observe that $\K$ is a component of $G-\V{\y}$, for any node $\y$, if and only if $\V{\y} \cap \K = \emptyset$ and $N(\K) \subseteq \V{\y}$.
		Since $\y_2 \notin M_\K \cup \Org_K$, there must be a vertex $u\in N(\K)\setminus\V{\y_2}$.
		On the other hand $u\in\V{\y_3}$ because $u\in N(\K)$.
		Finally, since $u\in N(\K)$ its model $M_u$ must intersect $M_\K$ and thus $u\in\V{\y_1}$,
		which is a contradiction to the connectivity of $M_u$.
	\renewcommand{\qed}{\hfill$\lhd$}\end{proof}

	Thus a \( \y \)-\remote component is also component of either \( \KK{x} \) or \( \KK{\z} \) but not both.
	Further, no component $\K$ is $\y$-\remote and $\y'$-\remote for distinct subdivision nodes \( \y,\y' \in \Tsub(x,\z) \).
	Indeed \( \Org_K \cap \Tsub[x,\z] \) is a connected sub-path with end nodes \( \y \) and one of \( x,\z \).
	Since the sub-path cannot have end node \( \y' \notin \{\y,x,\z\} \), the component \( \K \) is not \( \y' \)-\remote.

	Therefore instead of bounding the not-surrounded subdivision nodes $\y\in\Tsub(x,z)$ we may instead bound the $\y$-\remote components these nodes $y$.
	Then again each $\y$-\remote component is either a component  \( \KK{x} \) or \( \KK{\z} \) but not both.
	Each component \( \K \in \KK{x} \setminus \KK{\z} \) contains an \eye and each \eye is contained in at most one component of \( \KK{x} \), as seen in \autoref{lemma:K:of:m:partition}.
	Further a component \( \K \in \KK{x} \setminus \KK{\z} \) contains only \eyes of subtrees of \( x \) that do not contain \( \z \).
	For \( \K \in \KK{z} \setminus \KK{x} \) we have the analogous observations.
	Thus are at most as many not-surrounded \( \y \in \Tsub(x,\z) \) as there are \eyes unequal to \( x,\z \) in \( \Tsub \), which is \( |E(\Tsub)|-1 \).
\end{proof}

\newcommand{\RRs}{{\mathrm{R}}}
\newcommand{\RR}[1]{{\RRs(#1)}}

\newcommand{\gr}{{\bar{r}}} %
\newcommand{\unsur}{\overline{\mathcal{S}}}
\newcommand{\freeset}{F}
\newcommand{\freesetr}{F_r}
\newcommand{\TT}{\mathcal{T}}
\newcommand{\ambig}{\overline{\mathcal{U}}}
\newcommand{\tuple}{{B}}
\newcommand{\pot}{\Phi}
\newcommand{\EU}{E^{\uparrow}}
\newcommand{\free}{free\xspace}

\newcommand{\chain}{chain\xspace}

\newcommand{\treeopss}[1]{\rho^{\star}[{#1}]}
\newcommand{\template}{template\xspace}
\newcommand{\treeop}[1]{R{[#1]}} %
\newcommand{\treeops}[1]{\rho^\downarrow{[#1]}} %
\newcommand{\treeopsss}[1]{\rho{[#1]}} %
\newcommand{\ybar}{{{z}}}

\newcommand{\HH}{\mathcal{H}} %
\newcommand{\inner}{\mathcal{I}}
\newcommand{\innerr}{I}

\newcommand{\Yc}{\Gamma}

\subsection{Chains: Paths in any Representation}
\label{subsection:chains}

	\newcommand{\verteq}{\rotatebox{90}{$\,=$}}
	\newcommand{\vertdots}{\rotatebox{90}{$\,\dots$}}
	\newcommand{\vlangle}{\rotatebox{270}{$\langle$}}
	\newcommand{\vrangle}{\rotatebox{270}{$\rangle$}}
	\newcommand{\equalto}[2]{\underset{\let\scriptstyle\textstyle\substack{{\mkern4mu\verteq}\\{#2}}}{#1}} %

As observed previously, singleton guards $\{\ell\}$ and $\{r\}$ of a node $\y$ must neighbor $\y$.
If a path of nodes $\y_1,\dots,\y_s$ is made of aligned guards, i.e., $\{\y_{i-1}\}$ and $\{\y_{i+1}\}$ are the guards of $\y_i$, then it is a path in any representation $\Tsub$.
In this subsection we define such paths as \define{chains}.
A \define{chain} captures a maximum length path $\y_1,\dots,\y_s$ with aligned guards.
They also include the initial and final guard $\Y_0$ and $Y_{s+1}$, their \define{terminals}.

\begin{definition}
	A \define{\chain} is a maximal length $s\geq1$ sequence of sets of non-leaf nodes
	\[
		\YY = \langle \Y_0,\{\y_1\},\dots,\{\y_s\},\Y_{s+1} \rangle
	\]
	where ${\y_i}$ has guards $\Y_{i-1}$ and $\Y_{i+1}$ for every $i\in[s]$; and
		where $\Y_i \coloneqq \{\y_i\}$ for $i\in[s]$.
\end{definition}

To avoid lengthy statements, let us implicitly use $\Y_i \coloneqq \{\y_i\}$ from now on.
Let $\innerr(\YY)=\{\y_1,\dots,\y_s\}$ be the set of inner nodes of a chain $\YY$.
Let $\HH(G)$ be the set of chains of a (connected) graph $G$.

By \autoref{lemma:surround:characterization} such a chain implies that $\y_1,\dots,\y_s$ is a path in any representation $\Tsub$.
Also there are unique realizations of the terminals $\y_0\in\Y_0\cap N_{\Tsub}(\y_1)$ and $\y_{s+1}\in\Y_{s+t}\cap N_{\Tsub}(\y_{s})$ in a given $\Tsub$.
Let $\y_0\y_1\dots\y_s\y_{s+1}$ be the \define{corresponding} path of $\YY$ in the tree $\Tsub$.

The terminals define how the chains may attach to each other.
In the very simple case a terminal $Y_0$ may only consist of a single non-surrounded node, and hence any other chain must attach to that node.
Otherwise, as we show later, only the following option remains:
Terminal $Y_0$ contains some surrounded node $\y_i'$ which is part of another chain $\langle \Y_0', \dots, \Y_{s'+1}' \rangle$.
Interestingly $Y_0$ then contains the whole path $\y_1',\dots,\y_s'$.
This allows us to freely change the attachment of the chain $\langle \Y_0, \{\y_1\},\dots\rangle$ to the chain $\langle \dots,\{\y_i'\},\dots \rangle$ without breaking the connectivity of the representation, though possibly the properness.
Similarly, the intersection $\V{x} \cap \V{\y_i}$ for an outside-of-path node $x$ is equal for every $i\in[s]$, and therefore may be reattached in the same sense.

\begin{lemma}\sv{[$(\star)$]}
	\label{lemma:rehang:possible}
	Consider a chain $\langle Y_0,\dots,\{\y_i\},\dots,Y_{s+1}\rangle$.
	A neighbor $x \in N(\y_i)\setminus(\Y_{i-1}\cup\Y_{i+1})$
		has $\V{x} \cap \V{\y_i} = \V{x} \cap \V{\y_j}$, for every $j \in [s]$.
	Furthermore, $Y'\cap \innerr(\YY) \in \{\emptyset, \innerr(\YY)\}$ for every terminal $\Y'$ of any chain. %
\end{lemma}
\begin{proof}
	Clearly the lemma is true with the restriction of $i=j$ and assuming $\innerr(\YY)=\{\y_i\}$.
	We show that $\V{x} \cap \V{\y_i} = \V{x} \cap \V{\y_{i+1}}$
		and if for some $x'$ triple $(x',x,\y_i)$ is surrounding, also $(x',x,\y_{i+1})$ is surrounding;
		in other words, the statement for neighboring nodes \( \y_i \) and \( \y_j = \y_{i+1} \) where $\y_{s+1}$ is the unique node in $N_{\Tsub}(\y_s)\cap \Y_{s+1}$.
	By the transitivity of the claim and the symmetry of the chain, the general statement follows inductively.
	(Notably this proves an even stronger statement:
	A non-empty intersection $\Y'\cap \innerr(\YY)$ implies that also $\y_0,\y_{s+1} \in Y'$
		where $\y_{0}$ is the unique node in $N_{\Tsub}(\y_1)\cap \Y_{0}$.)

	Because \( (x,\y_i,\y_{i+1}) \) is \( \Tsub \)-ordered, it follows that \( \V{x} \cap \V{\y_i} \subseteq \V{x} \cap \V{\y_{i+1}} \).
	It remains to show the reverse subset relation.
	If \( \V{x} = \emptyset \), clearly also \( \V{x} \cap \V{\y_{i+1}} \subseteq \V{x} \cap \V{\y_{i}} \).
	Otherwise let \( \K \) be the component of \( \KK{\y_i} \) that contains the non-empty \( \V{x} \setminus \V{\y_i} \).
	Note that $\K$ is a third component of $\KK{\y_i}$ since unequal to $K_{\y_{i-1}},K_{\y_{i+1}}$ (where $\y_0\in Y_0$ is the node that neighbors $\y_1$).
	Thus condition~\ref{it:surround:subset} applies for the triple $(\y_{i-1},\y_i,\y_{i+1})$.
	It follows that \( N(\K) \subseteq \V{\y_{i-1}}\cap\V{\y_{i+1}} \), which makes \( \K \) in particular a component of \( \KK{\y_{i+1}} \).

	If $\V{x} \cap \V{\y_{i+1}} \not\subseteq \V{x} \cap \V{\y_{i}}$, there must be a vertex \( u \in  \V{x} \cap \V{\y_i} \setminus \V{\y_{i+1}} \).
	Then $\K$ as a component of $\KK{\y_{i+1}}$ must contain vertex $u$.
	Because of condition~\ref{it:surround:V:y} also $u\in \V{\y_{i-1}}$,
		and hence $u$ neighbors at least one vertex $v \in \V{\y_{i-1}} \setminus \V{\y_{i}}$.
	Thus $\K$ also contains vertex $v$.
	A contradiction since $\K$ as a component of $\KK{\y_{i}}$ cannot contain $v$ as it is from another subtree of $\y_i$.
	Therefore $\V{x} \cap \V{\y_{i+1}} \subseteq \V{x} \cap \V{\y_{i}}$

	Finally, because of the equal intersections \( \V{x} \cap \V{\y_i} = \V{x} \cap \V{\y_{i+1}} \) and since $(x',x,\y_i)$ is surrounding, also
		$(x',x,\y_{i+1})$ is surrounding.
\end{proof}

In the following subsection we aim for a bound on the number of chains.
We note here that chains behave in a reasonable way:
Each surrounded node $\y$ is part of exactly one chain, because otherwise it contradicts the classification by $\y$-\guards as seen in \autoref{lemma:surround:characterization}.
Clearly, a chain does not contain a node more than once, since $\y_i$-guards $Y_{i-1},Y_{i+1}$ are in different subtrees of $\y_i$, for every $i\in[s]$.
As the next step, we observe that terminals only consist of either a not-surrounded node or a non-singleton guard, i.e., are from
\begin{itemize}
	\item \( \unsur(G) \coloneqq \left\{ \{\y_0\} \mid \y_0 \in \CC(G) \text{ is not surrounded} \right\} , \) or
	\item \( \ambig(G) \coloneqq \left\{ Y_0 \mid Y_0 \text{ is guard of some } \y_1 \in \CC(G), |Y_0| > 1 \right\} . \)
\end{itemize}

\begin{lemma}\sv{[$(\star)$]}
\label{lemma:terminals}
Let $\Tsub$ be a \compact representation of a connected graph $G$.
Then every terminal $\Y_0,\Y_{s+1}$ of a chain of $G$ is part of\/ $\unsur(G)$ or $\ambig(G)$.
\end{lemma}
\begin{proof}
It suffices to consider $\Y_{s+1}$ due to the symmetry of chains.
If a terminal $\Y_{s+1}$ of a chain contains more then one maximal clique, then by definition $\Y_{s+1}\in\ambig(G)$.
Thus it remains to show that $\Y_{s+1}=\{\y_{s+1}\}$ implies that $\y_{s+1}$ is not surrounded.

Consider a chain $\langle\dots,\Y_{s-1},\{\y_s\},\{\y_{s+1}\}\rangle$.
Assume, for the sake of contradiction, that $\y_{s+1}$ is surrounded,
	and let $\Y_{s}'$ and $\Y_{s+2}$ be its guards.
By symmetry we may assume that $\y_{s}\notin\Y_{s+2}$.
Node $\y_{s+1}$ as a singleton $\y_s$-\guard neighbors $\y_s$.
Similarly there is a node $\y_{s-1} \in Y_{s-1}$ that neighbors $\y_s$,
	and there is a node $\y_{s+2} \in Y_{s+2}$ that neighbors $\y_{s+1}$.
Recall that for some node sets $L,R$ to be some $\y$-\guard, they must be separated by node $\y$.
Thus $\y_{s-1},\y_s,\y_{s+1},\y_{s+2}$ is a path in $\Tsub$.

In the following, we show that $\Y_{s}'=\{\y_s\}$, which contradicts that $\{\y_{s+1}\}$ is a terminal of this chain.
Our steps are (1) that $\y_s\in Y_{s'}$, then (2) that $Y_{s-1} \cap \Y_{s'} = \emptyset$, and finally (3) that $\{\y_s\}=Y_{s'}$.

(1)
We show that $\y_s\in Y_{s}'$.
Node $\y_{s+1}$ is contained in some chain $\langle \dots, \Y_{s}',\allowbreak\{\y_{s+1}\},\Y_{s+2},\allowbreak \dots \rangle$.
Assuming $\y_s\notin Y_{s}'$, makes $\y_s$ not part of this chain while neighboring $\y_{s+1}$, such that \autoref{lemma:rehang:possible} applies:
Because $(\y_{s-1},\y_s,\y_{s+1})$ is surrounding, also $(\y_{s-1},\y_s,\y_{s+2})$ is surrounding.
This contradicts that $\y_s$ has surroundings $\Y_{s-1}$ and $\{\y_{s+1}\}$.

(2)
We show that $Y_{s-1} \cap \Y_{s}' = \emptyset$.
Assume, for the sake of contradiction, that $Y_{s-1} \cap \Y_{s}' \neq \emptyset$.
Because $Y_{s'}$ as a $\y_{s+1}$-guard is connected, at least $\y_{s-1}\in Y_{s'}$.
Thus $\y_{s-1}$ and $\y_s'$ satisfy the minimality condition~\ref{it:surround:minimality},
	which implies that $\V{\y_{s-1}}\cap\V{\y_{s+1}} = \V{\y_s'}\cap\V{\y_{s+1}}$.
Considering a vertex $u \in \V{\y_s'} \setminus \V{\y_{s-1}}$, implies that $u \notin \V{\y_{s+1}}$.
Thus triple $(\y_{s-1},\y_s,\y_{s+1})$ is surrounding while $\V{y_s}$ has a `private' vertex $u\notin\V{\y_{s-1}}\cup\V{\y_{s+2}}$.
This implies $|\KK{\y_s}|=2$ as otherwise surround condition \ref{it:surround:V:y} fails.
Then surround condition \ref{it:surround:2} requires that $\V{\y_{s-1}}\cap\V{\y_{s+1}}=\emptyset$.
Now consider a vertex $v \in \V{\y_{s}} \cap \V{\y_{s+1}}$.
Then $v \notin \V{\y_{s-1}}$.
This however is a contradiction to $\V{\y_{s-1}}\cap\V{\y_{s+1}} = \V{\y_s'}\cap\V{\y_{s+1}}$, as observed before.

(3)
Finally, we show that $\Y_s'=\{\y_s\}$.
Assuming the contrary, there is a node $\y_s' \in Y_{s'}\setminus\{\y_s,\y_{s+1}\}$.
Because $Y_{s'}$ is connected, we may assume that $\y_s'$ neighbors $\y_s\in\Y_{s'}$.
Then $\y_s$ has at least three non-leaf neighbors $\y_{s-1}$, $\y_{s+1}$ and $\y_{s'}$,
	such that $\KK{\y_s}$ contains the three components $\K_{s-1},\K_{s+1}$ and $\K_{s}'$ containing the non-empty vertex sets of these respective nodes.
Now, consider a vertex $u\in\V{\y_s}\setminus\V{\y_{s+1}}$.
To satisfy condition~\ref{it:surround:V:y} for triple $(\y_{s-1},\y_s,\y_{s+1})$, we have $u \in \V{\y_{s+1}}$.
Note that $\V{\y_{s'}}\cap\V{\y_{s+1}}=\V{\y_{s}}\cap\V{\y_{s+1}}$ in order to satisfy the minimality condition~\ref{it:surround:minimality} surrounding $\y_{s+1}$.
Thus also $u\in\V{\y_{1'}}$.
Then however surround condition~\ref{it:surround:subset} for triple $(\y_{s-1},\y_{s},\y_{s+1})$ applies to component $K_{s'}$.
Hence also $u\in\V{\y_{s+1}}$, which contradicts our initial choice of $u\in\V{\y_s}\setminus\V{\y_{s+1}}$.
\end{proof}

Further, let the family of inner nodes be $\inner(G) \coloneqq  \left\{ \innerr(\YY) \mid \YY \in \HH(G) \right\}$.
Note here that $\inner(G)\cup\unsur(G)$ partition the maximal cliques $\CC(G)$.

\subsection{Template: Fixing the Topology of Chains}
\label{sec:template}

The set of chains~$\HH(G)$ of a (connected) graph~$G$ already considerably prescribes many paths that are present in any proper representation $\Tsub$ of~$G$.
What remains are two problems of a more global flavor:
For a chain there may be a vast range of possible connections.
Simultaneously we have to assure properness, i.e., that any vertices $u$ and $v$ escape each other.
To cope with these tasks we define a preliminary representation, a \define{template}.
A template considerably fixes the topology of a tree $\Tsub$ representing $G$.
It narrows down the possible representations such that we can focus on the properness.
At the same time, our final algorithm has to guess a template, thus its possibilities should be bounded by our parameter, the size of $T$.

To fix the relative positions of chains, a template locates the terminals of a chain, $Y_0$ and $Y_{s+1}$, on some template tree $\Ttemp$.
A concrete realization $\Tsub$ of that template is a subdivision of $\Ttemp$.
It realizes a chain between its terminals as prescribed by the template.
More precisely, $\ttemp$ maps the nodes $\ty{}$ of $\Ttemp$ to the terminals of chains.
	To avoid ambiguity, let $\ttemp(\ty{})$ not map to a mere terminal $Y_0$, if $Y_0\in\ambig(G)$ (as it may be huge), but narrow down the mapping to some set of inner nodes of $\inner(G)$.
	In other words we fix the neighborhood of a chain on the `chain-level'.
	Note that any $Y_0\in\ambig(G)$ is a superset of some set of inner nodes, as seen in \autoref{lemma:rehang:possible}.
For convenience, let us also fix a mapping $\htemp$ of chains $\langle\Y_0,\{\y_1\},\dots,\{\y_s\},\Y_{s+1}\rangle$ onto $\Ttemp$:
Let $\htemp$ map to paths $\ty{0},\dots,\ty{s'+1}$ in $\Ttemp$, which should be conforming with the terminals, which is $\ttemp(\ty{0})\subseteq Y_0$ and $\ttemp(\ty{s'+1})\subseteq Y_{s+1}$.
If the chain does not contain any branching node, it suffices to represent the terminals.
Then $\htemp$ maps simply to the single-edge path $\ty{0},\ty{s'+1}$ (for example $\YY'$ in \autoref{figure:template}(a),(b)).
In the other extreme, every inner node may be a branching node (respectively used as a terminal of another chain), thus possibly $s'= s$.

In more detail, a chain $\YY$ may correspond to a path  $\y_0\dots\y_{s+1}$ with an inner branching node $\y_0'$ which is the endpoint of a path $\y_0'\y_1'\dots$ corresponding to another chain $\YY' = \langle Y_0',\dots,\rangle$.
Thus, the chain $\YY$ must be mapped to a path with an inner node $\ty{j}$ that is an endpoint of the path $\ty{j},\typ{1} \dots, \typ{s''+1}$ that is the image of the other chain $\YY'$ (for example $\YY$ in \autoref{figure:template}a)b)).
As seen before, then $\innerr(\YY) \subseteq Y_0' \in \ambig(G)$.
Hence, the mapping of that terminal is $\ttemp(\ty{j})= \innerr(\YY)$.
We require this behavior for inner nodes like $\lambda_j$.

\begin{figure}[t]
\input{figure-template}
\label{figure:template}
\end{figure}

\begin{definition}
	Let \( G \) be a connected chordal graph.
	A \define{template} of a tree $T$ (w.r.t.\ $G$) is a triple $(\Ttemp,\ttemp,\htemp)$ where
	\begin{itemize}
		\item \( \Ttemp \) is a re-subdivision of \( T \),
		\item \( \ttemp \) is a mapping of the non-leaves of \( \Ttemp \) to \( \unsur(G) \cup \inner(G) \),
		\item \( \htemp \) is a bijection of the chains \( \HH(G) \) to an edge-disjoint set of non-trivial (i.e., containing at least one edge) paths between non-leaves of \( \Ttemp \), and
		\item for every chain $\langle Y_0,\dots,Y_{s+1}\rangle \in \HH(G)$ mapped to a path \( \ty{0},\dots,\ty{s'+1} \)
		we have that
		$\ttemp(\ty{0}) \subseteq Y_0$ and $\ttemp(\ty{s'+1}) \subseteq Y_{s+1}$
				and \( \ttemp(\ty{i}) = \innerr(\YY) \) for every $i \in [s']$.
	\end{itemize}
	Consider a tree $\Tsub$ where each non-leaf $\y$ is identified with a maximal clique $\V{\y}$.
	Then $\Tsub$ \define{realizes} the template \( (\Ttemp,\ttemp,\htemp) \) if
		$\Tsub$ results from subdividing $\Ttemp$ and
	\( \V{\ty{}} \in \ttemp(\ty{}) \) for every non-leaf $\ty{}$ of $\Ttemp$.
\end{definition}

Notably the image of $\htemp$ does not necessarily cover every edge between non-leaves.
Namely, non-surrounded nodes $y$ and $y'$ might be neighbors.
The next lemma establishes that every tree $\Tsub$ that is a \compact representation realizes some template, as we intended.

\begin{lemma}\sv{[$(\star)$]}
\label{lemma:template:exists}
If $\Tsub$ is a re-subdivision of a tree $T$ and a \compact representation of a connected graph \( G \), then
$\Tsub$ realizes some template $(\Ttemp,\ttemp,\htemp)$ of $T$.
\end{lemma}
\begin{proof}
Let \( \Ttemp \) be the tree \( \Tsub \) after contracting every not surrounded subdivision node.
Hence $\Tsub$ is a subdivision of \( \Ttemp \), and every triple of nodes $(\ty{1},\ty{2},\ty{3})$ of $\Ttemp$ is $\Ttemp$-ordered if and only if it is $\Tsub$-ordered.
Consider a chain $\langle Y_0,\dots,Y_{s+1} \rangle$ and its corresponding path $\y_0\y_1\dots,\y_s,\y_{s+1}$ in $\Tsub$.
Because each end node $\y_0$ and $\y_{s+1}$ is either not surrounded or is a branching node, they still exists in $\Ttemp$.
	Let \( \htemp \) map the chain \( \langle Y_0,\dots,Y_{s+1} \rangle \) to the non-trivial unique path in the tree \( \Ttemp \) between \( \y_0 \) and \( \y_{s+1} \).
	Note that each two sets of inner nodes $I_1,I_2$ of some distinct chain, hence $I_1,I_2\in\inner(G)$, are disjoint.
	Thus $\htemp$ maps each chain $\langle Y_0,\dots,Y_{s+1} \rangle$ to a contraction of the corresponding path, which makes \( \htemp \) a bijection to an edge-disjoint set of paths.
For every not surrounded non-leaf \( \y \in V(\Ttemp) \), let \( \ttemp(\y) = \{\y\} \in \unsur(G) \), which satisfies \( \V{\y} \in \ttemp(\y) \).
For every surrounded node \( \y_i \in V(\Tsub) \cap V(\Ttemp) \) and its chain \( \langle \cdot,\{y_1\},\dots,\{\y_s\},\cdot \rangle \), with \( 1 \leq i \leq s \),
	let \( \ttemp(\y_i) = \{\y_1,\dots,\y_s\} \in \inner(G) \), which also satisfies \( \V{\y_i} \in \ttemp(\y_i) \).

It remains to show  for every chain \( \langle Y_0,Y_1=\{\y_1\},\dots,Y_s=\{\y_s\},Y_{s+1}\rangle \) is mapped to a path \( \ty{0},\dots,\ty{k+1} \) in the tree \( \Tsub \)
	such that \( \ttemp(\ty{0}) \subseteq Y_0 \) and \( \ttemp(\ty{k}) \subseteq Y_{s+1} \) as well as \( \ttemp(\ty{i}) = \{y_1,\dots,\y_s\} \) for \( i \in [k] \).
By our construction \( \ttemp(\ty{i}) = \{\y_1,\dots,\y_s\} \) for \( i \in [k] \) as desired.
Let \( \y_0\y_1\dots\y_s\y_{s+1} \) be the corresponding path of \( \{\Y_0,\dots,\Y_{s+1}\} \) in \( \Ttemp \).
If \( Y_0 \in \unsur(G) \), then \( Y_0 = \{\y_0\} \) and hence we have that \( \ttemp(\ty{0}) = \y_0 \in Y_0 \) as desired.
In the remaining case \( Y_0 \in \ambig(G) \) which implies another chain \( \langle Y_0',\dots,Y_{s'+1}'\rangle \) where \( \y_{0} = y_{i'} \) for some \( i' \in [s'] \).
Let \( \y_0' \dots \y_{s'+1}' \) be its corresponding path.
Then \( \y_1 \) has guards \( Y_0 \) and \( Y_2 \) where \( Y_2 \) is a superset of \( \y_2 \) and \( Y_0 \) is a superset of \( \{\y_0',\y_1',\dots,y_{s'+1}'\} \) as seen in \autoref{lemma:rehang:possible}.
Thus we have \( \ttemp(\ty{0}) = \{\y_1',\dots,\y_{s'}'\} \subseteq Y_0 \) as desired.
Symmetrically it follows that \( \ttemp(X_{s+1}) \subseteq Y_{s+1} \).
\end{proof}

As formalized in the next lemma, we can enumerate the possible templates in \cFPT, since the number of chains is quadratically bounded in $V(T)$.

\begin{lemma}
\label{lemma:try:templates}
There are $2^{\Oh(t^2 \log t)} $ possible templates of a tree $T$ w.r.t.\ a connected chordal graph~$G$, which can be enumerated in time $2^{\Oh(t^2 \log t)} \cdot n^3$; where $t=|V(T)|$, $n=|V(G)|$.
\end{lemma}
\begin{proof}
Recall that the algorithm \ref{lemma:surround:characterization} from computes the guards $L$ and $R$ for every node $\y$ in time $\Oh(t^3n^3)$ where $t=|V(T)|$ and $n=|V(G)|$.
Upon discovering that $\y_1$ has a guard $Y_2$ let us link $\{\y_1\}$ to $Y_2$.
This way obtain the set of chains $\HH$ as double linked list with single links to the terminals.
At the same time we have the not surrounded sets $\{\y\}\in\unsur(G)$ as the set of $\y$ with empty guards.
It remains to try all possible mappings of $\unsur(G)$ and $\HH$ onto a re-subdivision $\Ttemp$.

How many chains are there?
Assume a yes-instance, hence some re-subdivision $\Tsub$ of $T$ is a compact representation of $G$.
Consider the chains where at least one terminal is realized by a branching node.
A branching node can only realize the terminal of at most its degree many chains.
Thus there are at most $2|E(T)|$ many chains where at least one terminal is realized by a branching node.
The remaining chains have terminals that are realized by subdivision nodes,
	hence have both terminals in $\unsur(G)$ as implied by \autoref{lemma:surround:characterization}.
Then every subdivision node of $\unsur(G)$ can be the terminal of at most two chains,
	which yields the bound of $|\unsur(G)|$ on the remaining chains.
In total we obtain the bound of $t+|\unsur(G)| = \Oh(t^2)$ on the number of chains.

To try all mappings of chains onto possibly $\Ttemp$, let us try all mappings to a tree $T^{-1}$ which is $T$ with sufficiently many subdivisions.
Since there are at most $\Oh(t^2)$ different chains and not-surrounded nodes it suffices to replace every edge of $T$ by a path of length $\Oh(t^2)$.
Now try all mappings of chains $\HH$ and not surrounded $\unsur(G)$ onto this super-preliminary tree $T^{-1}$, which are $\Oh(t^2)^{\Oh(t^2)} = 2^{\Oh(t^2 \log t)}$ mappings to try.
Then for each mapping, contract every path of unused inner subdivision nodes, which results in a template and is easily possible in time $\Oh(n^3)$.
\end{proof}

\subsection{Normalized Representation: Achieving Properness}\label{sec:norm}

We now consider a fixed template and  focus on the properness.
The remaining leeway is to locally change the branching nodes of a particular chain.
We use a construction which only fails if the considered template does not allow a \compact representation.
The result is a \define{normalized} representation.

Any representation $\Tsub$ can be normalized by a bottom-up process:
Move each branching node $\y_i$ up as much as possible within the local subtree, i.e.\ as long as the subtree remains compact.
By moving up, we mean replacing $\y_i$ by $\y_j$ as a branching node that is closer to a global root in the chain.
The set of nodes that potentially replace $\y_i$ behave in a linear fashion, and hence allow this greedy approach.
Thus we may assume that a normalized representation exists for a yes-instance.
Our algorithm though has to construct a representation from scratch.
By incorporating this idea in a more careful manner we may assemble each subtree of a normalized $\Tsub$ bottom-up.
Here we attach the inductive subtrees in the most conservative way;
then the same normalization step as before yields the desired new subtree.
Again, the linear behavior of the potential replacements enable this greedy approach.

To start, let us define the root of a template.
Since chains may not `align' towards a picked root $\gr_1$, we have to work with additional tie-breakers $\gr_2,\gr_3,\dots$.
\begin{definition}
A \define{root-ordering} $\gr$ is an ordering $\gr_1,\gr_2,\dots$ of nodes $V(\Ttemp) \cap \{ \ty{} \mid \ttemp(\ty{}) \in \unsur \}$.
\end{definition}

The specific root-ordering will not be of importance and we may pick one arbitrarily.
We assume in the following that every tree and template comes with a root-ordering.
See \autoref{figure:template} for an example.

\begin{definition}
A root-ordering $\gr$ and a template $(\Ttemp,\ttemp,\htemp)$ define an orientation for every chain $\YY=\langle Y_0,\dots,Y_{s+1} \rangle$ as follows.
Let $\htemp(\YY)$ map to a path in $\Ttemp$ with end nodes $\ty{0}$ and $\ty{s+1}$ where $\ttemp(\ty{0})\subseteq \Y_0$ and $\ttemp(\ty{s+1}) \subseteq \Y_{s+1}$.
Let $k$ be the smallest index such that $(\ty{0},\ty{s+1},\gr_k)$ or $(\gr_k,\ty{0},\ty{s+1})$ is $\Ttemp$-ordered.
If $(\ty{0},\ty{s+1},\gr_k)$  is $\Ttemp$-ordered, then \define{$\YY$ is oriented towards $\Y_{s+1}$}, which we denote by writing $\langle Y_0,\dots,Y_{s+1} \rangle^\gr$.
\end{definition}

Note that the index $k$ always exists, since every neighbor of a leaf is not surrounded.

Let $\treeop{\y_i,\y_j}\Tsub$ be the tree resulting from replacing branching node $\y_i$ by $\y_j$.
Its local version is $\treeopsss{\y_i,\y_j}\Tsub$.
The models living in the more restrict subtree $\treeops{\y_i,\y_j}\Tsub$ are critical:
	Their properness is at stake.
We define the possible replacements of a node $\y_i$ resulting in a proper representation as the potential $\pot(\Tsub,\y_i)$.

\begin{definition}
Let~$\gr$ be a root-ordering.
Consider a branching node \( y_i \in V(\Tsub) \) and its chain
$ \langle Y_0,\dots,\{\y_i\},\dots,\{\y_j\},\dots,Y_{s+1} \rangle^\gr $ where $\y_0$ realizes $\Y_0$.
For integers $i\leq j<s$, let
\begin{itemize}
	\item \( \treeop{\y_i,\y_j}\Tsub \) be the tree \( \Tsub \) where $\y_i$ replaces $\y_j$ as a branching node,
			i.e., edge $\{\y_i,z\}$ is replaced by a new edge $\{\y_j,z\}$, for every node \( z \in N_{\Tsub}(\y_i) \setminus (\Y_{i-1}\cup\Y_{y+1}) \);
	\item \( \treeopsss{\y_i,\y_j}\Tsub \) be the tree consisting of the subtree of $\treeop{\y_i,\y_j}\Tsub$ rooted at $\y_0$ (w.r.t.\ global root $\gr_1$) and path $\y_0,\dots,\y_s$ where for every chain node $\y_{i'}\in\{\y_1,\dots,\y_s\}$ and non-chain neighbor $z' \in N_{\treeop{\y_i,\y_j}\Tsub} \setminus (Y_{i'-1} \cap Y_{i'+1})$ a new leaf node $\y_{i'}'$ added adjacent to $\y_{i'}$;
	\item \( \treeops{\y_i,\y_j}\Tsub \) be the tree consisting of the subtree of $\treeop{\y_i,\y_j}\Tsub$ rooted at $\y_0$ (w.r.t.\ global root $\gr_1$) and path $\y_0,\dots,\y_{j-1}$.
\end{itemize}
\end{definition}

For convenience, let \( \treeops{\y_i}\Tsub \coloneqq \treeops{\y_i,\y_i}\Tsub \) as well as \( \treeopsss{\y_i}\Tsub \coloneqq \treeopsss{\y_i,\y_i}\Tsub \).

\begin{definition}
	We define the \define{potential} $\pot(\Tsub,\y_i)$ (w.r.t.\ a template $(\Ttemp,\ttemp,\htemp)$ and root-ordering $\gr$) of non-leaf node \( \y_i \).
	\begin{itemize}
		\item For a not-surrounded node \( \y_i \), let \( \pot(\Tsub,\y_i) = \{\y_i\} \).
		\item For a surrounded branching node \( \y_i \), consider its chain \( \langle\cdot, \dots,\{\y_i\},\dots,\{\y_s\}, \cdot \rangle^\gr \).
		The potential $\pot(\Tsub,\y_i)$ contains every node $y_j \in \{ \y_i,\dots,\y_s \}$
		where the tree $\treeop{\y_i,\y_j}\Tsub$ is such that
		every vertex $u$ with model $M_u \subseteq V(\treeops{\y_i,\y_j}\Tsub)$ escapes every other vertex $v$.
	\end{itemize}
\end{definition}

A simple example is that $\y_i\in\pot(\Tsub,\y_i)$ for \compact representations $\Tsub$, as the considered replacement does nothing.
In contrast, $\pot(\Tsub,\y_i)=\emptyset$ indicates non-properness for the subtree of $\y_{i-1}$.
Indeed, the potential of $\y_i$ captures exactly the possible replacements of $\y_i$ as a branching node.

If some replacement $\y_j$ of $\y_i$ already is a branching node, the topology changes and $\treeop{\y_i,\y_j}\Tsub$ does not realize the same template.
To avoid such issues, we require (without loss of generality) a minimal representation:
	A tree \( \Tsub \) is \define{minimal} if
		there is no \compact representation \( \Tsub' \) of \( G \) that is a re-subdivision of \( \Tsub \) with fewer branching nodes.
Clearly, if there is a representation $\Tsub$ of $G$, we may also assume that it is minimal.
In particular, the contraction would result in different candidate re-subdivision of $T$, which we consider separately.

We may compute it locally, meaning it suffices to consider the subtree $\treeopsss{\y_i}$.
Since the potential $\pot(\Tsub,\y_i)$ is a connected subsequence of $\langle \y_i,\dots,\y_s \rangle$, we either view it as a set or as such a subsequence.
Further, if the potential is $\langle \y_i,\dots,\y_j\rangle$, then replacing $\y_i$ with the last node $\y_j$ makes the resulting potential at $\y_j$ singleton.
Finally, the potential is independent from later replacements, assuming a bottom-up (i.e., leaf-to-root) procedure.

\begin{lemma} %
	\label{lemma:pot:is:connected}
Let $\Tsub$ be a minimal \compact representation of a connected graph $G$.
We observe the following for a chain $\langle\cdot,\dots,\{\y_i\},\dots,\{\y_j\},\dots,\{\y_s\},\cdot\rangle^\gr$ for $i\leq j \leq s$:
\begin{enumerate}
\item
If $\y_j\in\pot(\Tsub,\y_j)$, then $\treeop{\y_i,\y_j}\Tsub$ is a minimal \compact representation of $G$.
\item
locality, $\pot(\Tsub,\y_i) = \pot(\treeopsss{\y_i}\Tsub,\y_i)$,
\item
connectivity,
$\pot(\Tsub,\y_i)$ is connected in $\Tsub$, and hence some subsequence $\langle\y_i,\dots,\y_j\rangle$,
\item
linearity,
$\pot(\Tsub,\y_i)=\langle\y_i,\dots,\y_j\rangle$ if and only if $\pot(\treeop{\y_i,\y_j}\Tsub,\y_j) = \langle \y_j \rangle$.
\item
independence,
$\pot(\treeop{\y_i,\y_j}\Tsub,x) \subseteq \pot(\Tsub,x)$ for every node $x \in V(\Tsub)$ where $(x,\y_i, \gr_1)$ is $\Tsub$-ordered.
\end{enumerate}
\end{lemma}
\begin{proof}
In the following we show the claims (1) to (5) with some intermediate steps.
The minimality of the representation $\Tsub$ is only needed to show that the modified tree $\treeop{\y_i,\y_j}\Tsub$ remains minimal, as part of (1), and for (2), locality.

\medskip

(1)
Every node is a leaf in \( \Tsub \) if and only if it is a leaf in \( \treeop{\y_i,\y_j}\Tsub \).
Thus the representation still has unique maximal cliques at every non-leaf node, and no model contains a leaf.
By \autoref{lemma:rehang:possible} we have that \( \V{z} \cap \V{y_i} = \V{z} \cap \V{y_j} \) for every non-chain child \( z \in N_{\Tsub}(\y_i) \setminus (\Y_{i-1} \cup \Y_{i+1}) \).
Thus every model still induces a connected subgraph in \( \treeop{\y_i,\y_j}\Tsub \).

It remains to show that every vertex \( u \in V(G) \) \escapes every vertex \( v \in V(G) \) at some edge in \( E(\treeop{\y_i,\y_j}\Tsub) \).
Let vertex \( u  \) strongly \escapes \( v \) at some edge \( e \in E(\Tsub) \).
If also \( e \in E(\treeop{\y_i,\y_j}\Tsub) \), then clearly \( u \) \escapes \( v \) at \( e \) also in the tree \( \treeop{\y_i,\y_j}\Tsub \).
Otherwise, we have either the case that \( e= (\y_i,\ybar) \) or \( e=(\ybar,\y_i) \) for some non-chain child \( \ybar \in N_{\Tsub}(\y_i) \setminus (\Y_{i-1} \cup \Y_{i+1}) \).
\begin{itemize}
	\item
	Consider that \( u \) strongly \escapes \( v \) at edge \( (\ybar,\y_i) \).
	By \autoref{lemma:rehang:possible} we have that \( \V{\ybar} \cap \V{\y_i} = \V{\ybar} \cap \V{\y_j} \).
	Thus still \( u,v \in \V{\ybar} \) while we have that \( v \notin \V{v_j} \) since \( v \notin \V{v_i} \).
	Hence \( u \) \escapes \( v \) at \( (\ybar,\y_j) \).
	\item
	Else \( u \) strongly \escapes \( v \) at edge \( (\y_i,\ybar) \).
	Then \( u,v \in \V{\y_i} \) and \( v \notin \V{\ybar} \).
	If vertex \( u \in \V{\y_j} \), then \( u \) \escapes \( v \) at edge \( (\y_j,\ybar) \).
	Thus consider that \( u \notin \V{\y_j} \).
	Because \( u \in \V{\y_i} \setminus \V{\y_j} \), we have that its model \( M_u \subseteq V( \treeops{\y_i,\y_j}\Tsub ) \).
	Since \( \y_j \in \pot(\Tsub,\y_i) \), by definition \( u \in V( \treeops{\y_i,\y_j}\Tsub ) \) escapes \( v \) at some edge in \( \treeops{\y_i,\y_j}\Tsub \).
\end{itemize}
Thus \( \treeop{\y_i,\y_j}\Tsub \) is a \compact representation of $G$.
Note that this also implies the reflexivity $\y_i\in\pot(\Tsub,\y_i)$.

Now it remains to show that $\treeop{\y_i,\y_j}\Tsub$ is also \stress{minimal}.
It suffices to show that $\{\y_{i+1},\dots,\y_{j}\}$ contains no branching node.
Assuming otherwise, there is a first branching node $\y_k\in\{\y_{i+1},\dots,\y_{j}\}$, meaning every node of $\{\y_{i+1},\dots,\y_{k-1}\}$ is a subdivision node.
Then $\treeop{\y_i,\y_k}\Tsub$ is a \compact representation of $G$.
However, also $\treeop{\y_i,\y_k}\Tsub$ results from a series of subdivisions of $\Tsub$ where path $\y_i,\dots,\y_k$ is contracted; thus has one less branching node than $\Tsub$.
A contradiction to the minimality of $\Tsub$.

\medskip

(2)
We show that $\pot(\Tsub,\y_i) = \pot(\treeopsss{\y_i}\Tsub,\y_i)$.
Observe that the vertices that need to escape some other vertex in the modified trees  $\treeop{\y_i,\y_j}\Tsub$ and $\treeop{\y_i,\y_j}\treeopsss{\y_i}\Tsub$, for some $\y_j\in\{\y_i,\dots,\y_s\}$, are the same:
	They are defined via the subtrees $\treeops{\y_i,\y_j}\Tsub $ and $\treeops{\y_i,\y_j}\treeopsss{\y_i}\Tsub$,
	which are equal.
We continue to show both directions

Consider that $\y_j \in \pot(\Tsub,\y_i)$ for some node $\y_j\in\{\y_i,\dots,\y_s\}$.
Then every vertex $u \in \treeops{\y_i,\y_j}\Tsub$ escapes every vertex $v \in V(G)$.
Note that none of the vertices $\{\y_{i+1}, \dots, \y_j\}$ is a branching node in tree $\Tsub$ since a branching node would contradict minimality as discussed in (1).
We have to show that every vertex $u \in V(\treeops{\y_i,\y_j}\treeopsss{\y_i}\Tsub)$ escapes every vertex $v \in V(G)$.
Let $u$ escape $v$ at some edge $(x,y) \in E(\Tsub)$,
	which must be incident to a vertex from $V(\treeops{\y_i,\y_j}\Tsub)$.
If even $(x,y)\in E(\treeops{\y_i,\y_j}\Tsub)$, then because of our initial observation $u$ also escapes $v$ in tree $\treeopsss{\y_i}\Tsub$.
Thus it remains to consider that $x \in V(\treeops{\y_i,\y_j}\Tsub)$ and $y \notin V(\treeops{\y_i,\y_j}\Tsub)$.
Recall that none of the vertices $\{\y_{i+1}, \dots, \y_j\}$ is a branching node in tree $\Tsub$.
Then by construction of $\treeopsss{\y_i}\Tsub$, there is a new leaf $\y'$ neighboring $x$ with $\V{\y'}=\emptyset$.
Thus $u$ escapes $v$ at edge $(x,\y') \in E(\treeopsss{\y_i}\Tsub)$.

For the other direction, consider that $\y_j \in \pot(\treeopsss{\y_i}\Tsub)$ for some node $\y_j\in\{\y_i,\dots,\y_s\}$.
Then every vertex $u \in V(\treeops{\y_i,\y_j}\treeopsss{\y_i}\Tsub)$ escapes every other vertex $v \in V(G)$.
We have to show that every vertex $u \in V(\treeops{\y_i,\y_j}\Tsub)$ escapes every other vertex $v \in V(G)$.
Let $u$ escape $v$ at some edge $(x,y) \in E(\Tsub)$,
	which must be incident to a vertex from $V(\treeops{\y_i,\y_j}\Tsub)$.
If even $(x,y) \in E(\treeops{\y_i,\y_j}\Tsub)$, then clearly $u$ also escapes $v$ in tree $\Tsub$.
Thus it remains to consider that $x \in V(\treeops{\y_i,\y_j}\Tsub)$ and $y \notin V(\treeops{\y_i,\y_j}\Tsub)$.
Then by construction of $\treeopsss{\y_i}\Tsub$, node $\y$ must be newly added leaf to $x$ due to an edge $(x,\y') \in E(\Tsub)$ which was removed.
Thus $u$ escapes $v$ at edge $(x,\y') \in E(\Tsub)$.

\medskip

(Monotonicity) We have that $\pot(\treeop{\y_i,\y_j}\Tsub,\y_j) \subsetneq \pot(\Tsub,\y_i)$ for $i<j$.
By definition $\y_i \notin \pot(\treeop{\y_i,\y_j}\Tsub,\y_j$ which shows the inequality.
The subset relation follows from the following observation.
The set of vertices whose escape has to be guaranteed increases since the node set of the considered tree increases, which is $V(\treeops{\y_i}\Tsub)\subseteq V(\treeops{\y_i,\y_j}\Tsub)$,
hence the condition has to apply for more vertices.
Meanwhile the set of edges where an escape is possible remains equal $E(\treeopss{\y_i}\Tsub) = E(\treeopss{\y_i\,\y_j}\Tsub)$.

\medskip

$(4')$
As a first step for (4), we show that $j$, the maximal index such that $\y_j\in\pot(\treeop{\y_i,\y_j},\y_i)$, has $\pot(\treeop{\y_i,\y_j}\Tsub,\y_j) = \langle \y_j \rangle$.
Because of (1) tree $\treeop{\y_i\,\y_j}\Tsub$ is a \compact representation of $G$, and thus we have the reflexivity $\y_j\in\treeop{\y_i\,\y_j}\Tsub$.
Further monotinicity implies that $\pot(\treeop{\y_i,\y_j}\Tsub,\y_j)$ is a subsequence of $\{\y_i,\dots,\y_j\}$, and thus only consists of $\{\y_j\}$.

\medskip

(3)
Let $j$ be maximal such that $\y_j\in\pot(\treeop{\y_i,\y_j},\y_i)$.
Then because of $(4')$ we have that $\pot(\treeop{\y_i,\y_j}\Tsub,\y_j) =\{\y_j\}$.
Because of the monotinicity $\pot(\treeop{\y_i,\y_{j-1}}\Tsub,\y_{j-1})=\{\y_{j-1},\y_j\}\subsetneq\{\y_j\}$.
Inductively the monotinicity implies $ \pot(\Tsub,\y_i) = \pot(\treeop{\y_i,\y_{i}}\Tsub,\y_{i})=\{\y_{i},\dots,\y_{j}\}$.

\medskip

(4)
Because of $(4')$ and (3) we have the \forward-direction that if $\pot(\Tsub,\y_i)=\langle\y_{i},\dots,\y_{j'} \rangle$ then $\pot(\treeop{\y_i,\y_j}\Tsub,\y_j) = \langle \y_j \rangle$.
For the \backward-direction, assuming $\pot(\treeop{\y_i,\y_j}\Tsub,\y_j) = \langle \y_j \rangle$, implies due to monotinicity that $\pot(\Tsub,\y_i)=\langle \y_{i},\dots,\y_{j'} \rangle$ for an $j'\geq j$.
It remains to show that $j'=j$.
By reflexivity $y_{j'} \in \pot(\treeop{\y_i,\y_j}\Tsub,\y_{j'})$.
Then again by monotonicity $\y_{j'} \in \pot(\treeop{\y_i,\y_{j'}}\Tsub,\y_{j'}) \subseteq \pot(\treeop{\y_i,\y_j}\Tsub,\y_{j})$.
Since however $\treeop{\y_i,\y_j}\Tsub=\langle\y_j\rangle$ we have that $j'=j$ as desired.

\medskip

(5)
We prove (5) independently of conditions (1)-(4):
Consider a node $x=x_{i'}$ in a chain $\langle X_0,\dots,\{x_{i'}\},\dots,X_{s'+1}\rangle^\gr$ (possibly the same chain as of $\y_i$) and where $(x_{i'},\y_i,\gr_1)$ is $\Tsub$-ordered.
Assume that the rehang operation makes some node $x_{j'}$ appear in the potential of $x_{i'}$, which is
\[
	x_{j'} \in \pot(\treeop{y_i,y_j}\Tsub,x_{i'}) \setminus \pot(\Tsub,x_{i'}) .
\]

Then there is a vertex \( u \) that escapes some vertex \( v \) in the tree \( \treeop{x_{i'},x_{j'}}\treeop{y_i,y_j}\Tsub \) but not in the tree \( \treeop{x_{i'},x_{j'}}\Tsub \).
Thus \( u \) escapes \( v \) at edge \( (x_{i'},z') \) but does not escape at edge \( (x_{j'},z') \) for some \( z' \in N_{\Tsub}(x_{i'}) \setminus (X_{i'-1} \cup X_{i'+1}) \).
This means $\y_i \in M_u$ while $\y_j \notin M_u$.
Similarly, \( u \) escapes \( v \) at edge \( (\y_j,z) \) but does not escape at edge \( (\y_i,z) \) for some \( z \in N_{\Tsub}(\y_{i'}) \setminus (\Y_{i'-1} \cup \Y_{i'+1})  \).
That means $x_{i'}\in M_u$ while $x_{j'}\in M_u$.

Because of the connectivity of model $M_u$ then $(x_{i'},x_{j'},\y_i)$ and $(x_{i'},x_j,\y_i)$ is $\Tsub$-ordered.\footnote{Later in this work, we define minimality, under which we would be able to assume that $(x_{i'},x_{j'},\y_j,\y_i)$ is $\Tsub$-ordered.}
Recall that $(x_{i'},\y_i,\gr_1)$ is $\Tsub$-ordered.
Thus $(x_{i'},\y_j,\y_i,\gr)$ is $\Tsub$-ordered.
That makes $1$ the defining lowest index of the orientation of chain of $\y_i$, which we denote as ${}^\gr\langle \dots,\{\y_i\},\dots,\{\y_j\},\dots\rangle$.
A contradiction that $\y_j$ is in the potential $\pot(\Tsub,y_i)$.
\end{proof}

Consider a tree $\Tsub$ that realizes a template $(\Ttemp,\ttemp,\htemp)$, and has some root-ordering~$\gr$.
We say $\Tsub$ is \define{normalized} for a node \( y \) (w.r.t.\ to $(\Ttemp,\ttemp,\htemp)$ and $\gr$) if \( \pot(\Tsub,y) = \langle y \rangle \).
By the locality property, this is equivalent to $\pot(\treeopsss{y}\Tsub,y) = \langle \y \rangle$, hence it suffices to consider the local subtree.
	The whole tree \( \Tsub \) is normalized if it is normalized for every branching node.
Now the independence of the potential as explored earlier allows normalizing any representation by a bottom-up procedure.
Thus, in a yes-instance, we may assume a normalized representation.

\newcommand{\X}{X}
\newcommand{\x}{x}

\begin{lemma}\sv{[$(\star)$]}
	\label{lemma:realization:algorithm}
	There is an $\Oh(n^3)$ time algorithm that, given a connected chordal $n$-vertex graph~\( G \) and a template $(\Ttemp, \ttemp, \htemp)$,
		decides whether there is a minimal \compact representation of~\( G \) that realizes $(\Ttemp,\ttemp,\htemp)$,
		and if one exists, it outputs one that is also normalized. %
\end{lemma}
\begin{proof}
First, let us prove that we may assume a normalized representation for a yes-instance.

\begin{claim}
	\label{lemma:normalized}
Consider a tree $\Tsub$ that realizes a template $(\Ttemp,\ttemp,\htemp)$ and is a minimal \compact representation of a connected graph~\( G \).
	Then there is a tree \( \Tsub' \) that is a \underline{normalized} (w.r.t.\ to $(\Ttemp,\ttemp,\htemp)$ and to some root-ordering~$\gr$) minimal \compact representation of a graph \( G \) and realizes \( (\Ttemp,\ttemp,\htemp) \).
\end{claim}
\begin{proof}
	Fix a bottom-up ordering $\ty{1},\ty{2},\dots$ of the non-leaf nodes of the template tree \( \Ttemp \) where the children of \( \ty{k} \) (regarding the root $\gr_1$) occur before $\ty{k}$ for every \( k \geq 1 \).
	We prove by induction over \( k \geq 0 \),
		that there is a minimal \compact tree representation \( T_k^\star \) of \( G \) that is normalized for $\ty{1},\dots,\ty{k}$.
	Eventually this yields a normalized minimal \compact tree representation \( \Tsub' \) of a graph \( G \) that realizes template \( (\Ttemp,\ttemp,\htemp) \).

	For the induction base, tree \( T^\star_{0} \coloneqq \Tsub^\star \) is a \compact representation of \( G \) that realizes the template \( (\Ttemp,\ttemp,\htemp) \).
	No condition for nodes $\ty{1},\dots,\ty{k}$ has to be satisfied since \( k=0 \).

	For the induction step, consider \( T^\star_{{k-1}} \) from the induction hypothesis, which is a minimal \compact representation of \( G \) and normalized for $\ty{1},\dots,\ty{k-1}$.
	Let node $\y_i\coloneqq\ty{k}$ and consider its potential $\pot(T_{{k-1}}^\star, \y_i) =\allowbreak \langle \y_i,\dots,\y_j \rangle$ (possibly with $\y_i=\y_j$).
	Then by \autoref{lemma:pot:is:connected} tree $T_{k}^\star \coloneqq \treeop{\y_i,\y_j}T_{k-1}^\star$ is a minimal \compact representation which is normalized for the realizations of $\ty{1},\dots,\ty{k-1}$,
		assuming that $(\ty{k'},\ty{k},\gr_1)$ is $T_{k-1}^\star$-ordered for every node $\ty{k'}$ with $k'<k$.
	This is the case because of our bottom-up ordering of the non-leaves of $\Ttemp$ and because the relative position of the branching nodes is the same for every representation that realizes $(\Ttemp,\ttemp,\htemp)$.
	More precisely, since every representation $\Tsub'$ of $G$ that realizes $(\Ttemp,\ttemp,\htemp)$ is a subdivision of $\Ttemp$, every triple of nodes $(\ty{1},\ty{2},\ty{3})$ is $\Tsub'$-ordered if and only if it is $\Ttemp$-ordered.
	Finally $T_{k}^\star = \treeop{\y_i,\y_j}T_{k-1}^\star$ is also normalized for $\ty{k}$ since by linearity of the potential $\treeop{\y_i,\y_j}T_{k-1}^\star=\langle\y_j\rangle$ as seen in \autoref{lemma:pot:is:connected}.
\renewcommand{\qed}{\hfill$\lhd$}\end{proof}

Now we prove the algorithmic statement.
\newcommand{\ttarget}{\Tsub'}
Assuming a yes-instance, there is minimal \compact representation \( \ttarget \) of \( G \) that realizes template $(\Ttemp,\ttemp,\htemp)$.
By \autoref{lemma:normalized}, we may assume that \( \ttarget \) is normalized.
Our algorithm outputs a representation isomorphic to $\Tsub'$, thus a normalized one as desired.
If, however, our construction fails at some point, we correctly conclude that no such representation exists.

We fix an ordering $\sigma=\ty1,\ty2,\dots$ of the non-leaf nodes of the template tree $\Ttemp$,
	which follows the ordering within in a chain and otherwise is bottom-up.
Pick a node $\ty{k}$ where every non-leaf child of $\ty{k}$ has been added before, and append it to the ordering.
If there is a chain $\langle Y_0,\dots,Y_{s'+1} \rangle^\gr$ mapped by $\htemp$ to a path of form $\ty{0},\ty{k},\ty{k,1},\dots,\ty{k,s'+1}$, append nodes $\ty{k,1},\dots,\ty{k,s'+1}$ as well.
Then continue to picking a new node until all nodes are ordered.

For $k\geq1$, let $T_k$ be the subtree of $\Tsub'$ induced by $\ty{1},\dots,\ty{k}$, every subdivision node between nodes from $\ty{1},\dots,\ty{k}$ and leaves neighboring $\ty{1},\dots,\ty{k}$.
	Let $T_k'$ be the tree $T_k$ where the node $\y$ closest to root $\gr_1$ has an additional leaf node $\y'$ adjacent to $\y$ (where possibly $\y \neq \y_k$).
	Considering this tree $T_k'$ with an extra leaf node may be necessary when considering the potential $\pot$ where then vertices may escape at edge $(\y,\y')$.

By induction over $k \geq 1$, we prove that a tree isomorphic to $T_k$
	is polynomial time computable given \( G \), $(\Ttemp,\ttemp,\htemp,)$ and $\gr$.
Eventually this yields to a representation \( \Tsub \) isomorphic to \( \ttarget \), thus normalized minimal \compact and realizing $(\Ttemp,\ttemp,\htemp)$, as desired.

\medskip

(Induction base, when $\ty{k}$ neighbors a leaf (w.r.t.\ to root $\gr_1$))
Node $\ty{k}$ represents a not-surrounded node $\ttemp(\ty{k}) = \{\ty{k}\} \in \unsur(G)$.
Then \( T_k \) consists only of $\ty{k}$ adjacent to some leaf.
Thus this tree is prescribed by $(\Ttemp,\ttemp,\htemp)$ and hence no computation is required.

\medskip

(Induction step $\unsur(G)$)
Consider that $\ty{k} \in V(\Ttemp)$ is a not-surrounded node, which is $\ttemp(\ty{k}) = \{\ty{k}\} \in \unsur(G)$.
Let $\ty{c_1},\dots,\ty{c_z}$ be the children of $\ty{k}$ in the tree $\Ttemp$.
By the induction hypothesis, the subtrees $T_{{c_1}},\dots,T_{{c_z}}$ are polynomial time computable.
Tree $T_k$ must realize $\{\ty{k}\}$ with the only possibility being $\y_k$.
We show how to construct a tree isomorphic $T_k$, given fixed child subtrees.
After construction we validate that $\pot(T_k,\ty{k})=\langle \ty{k} \rangle$, i.e., that the constructed tree is indeed a \compact representation, in polynomial time.

Consider the adjacency of $\lambda_{c_1}$ and $\lambda_{k}$ in the template tree.
A simple case is that also $\ttemp(\lambda_{c_1})=\{\lambda_{c_1}\}\in\unsur(G)$.
Then in the tree $T_k$ the two realizing nodes $\lambda_{c_1},\lambda_{k}$ must be adjacent,
	and we define $\overrightarrow{c_1}\coloneqq\lambda_{c_1}\lambda_k$.

Otherwise there is a chain $\YY_1$ mapped to a path containing the edge $\lambda_{c_1}\lambda_k$.
More so, since $\ttemp(\lambda_k)\in\unsur(G)$ this chain has terminal $\{\lambda_k\}$, i.e.\ the chain has format $\langle \dots,\{\lambda_k\}\rangle^\gr$.
Here we may assume the orientation of the chain because of how the ordering $\sigma$ is defined.
Note that $\htemp(\YY_1)$ possibly only contains edge $\{\lambda_{c_1},\lambda_k\}$,
	in which case $\ttemp(\lambda_{c_1})$ is the other terminal of $\YY_1$;
	but only in this particular case.
We further distinguish:

If $\ttemp(\ty{c_1})$ is the other terminal of the chain $\YY_1$, it has format
$\langle Y_0', \{\y_{c_1}^1\},\dots, \{\y_{c_1}^{s_{1}}\}, \{\ty{k}\} \rangle^\gr$ where $\ttemp(\ty{c_1})\subseteq Y_0'$
(not necessarily equality since possibly $\ttemp(\lambda_{c_1})\in\inner(G)$ being subset of a larger guard.)

Otherwise, if $\ttemp(\ty{c_1})$ is not the other terminal of $\YY_1$,
then $\ttemp(\ty{c_1}) = \inner(\YY_1)$.
The root $\ty{c_1}\in\ttemp(\lambda_{c_1})$ of tree $T_{c_1}$ must be among these inner nodes.
Hence chain $\YY_1$ has format $\langle \cdot, \dots,\{\ty{c_1}\}, \{\y_{c_1}^1\},\dots, \{\y_{c_1}^{s_{1}}\}, \{\ty{k}\} \rangle^\gr$.

In both cases, path $\ty{c_1},\y_{c_1}^1,\dots, \y_{c_1}^{s_{1}},\ty{k}$ connects the realizations of $\lambda_{c_1}$ and $\lambda_k$, which we denote as $\overrightarrow{c_1}$.
That is the only possibility how $T_k$ connects the root of subtree $T_{c_1}$ with $\ty{k}$.
Let the paths \( \overrightarrow{c_2},\dots, \overrightarrow{c_z} \) be defined analogously.
Clearly, the same observations apply to these paths.
Then the subtree of $T_k'$ of $\ttarget$ rooted at $\ty{k}$ consists of subtrees $T_{{c_1}},\dots,T_{{c_z}}$ together with paths $\overrightarrow{c_1},\dots, \overrightarrow{c_z}$.

\medskip

(Induction step $\inner(G)$)
It remains the interesting case where template node $\ty{k}$ is a surrounded branching node.
This means that $\ttemp(\ty{k}) = \{\y_1,\dots,\y_s\}$ for some set of $s\geq1$ inner nodes from $\inner(G)$.
Here the realization of $\ty{k}$ without knowing $T_k$ is not immediately clear.
Let $\YY=\langle\Y_0,\{\y_1\},\dots,\{\y_s\},\Y_{s+1}\rangle^\gr$ be the unique chain containing these inner nodes.
The template maps $\YY$ to a non-trivial path $\ty{0},\dots,\ty{s'+1}$ in $\Ttemp$ containing $\ty{k}$.
Thus we have the following situation:
\[
	(\ty{0},\dots,\ty{c_0},\ty{k},\ty{c_0'},\dots,\ty{s'+1})
	\; = \;
	\htemp\big(\langle Y_0,\{\y_1\} \dots,\{\y_i\},\dots,\{\y_s\}
	,Y_{s+1} \rangle^\gr\big) .
\]
Every of those inner template nodes $\lambda_{i'}$ has $\ttemp(\lambda_{i'})=\{\y_1,\dots,\y_s\}$.
Let us assume that $\ttemp(\ty{0})\subseteq Y_0$, such that the directions of increasing indexes match.

Note that $\ty{c_0}$ is ordered before $\ty{k}$ because of how the ordering $\sigma$ is defined.
Our algorithm may determine $\ty{c_0}$ as the child in $\Ttemp$ where $(\ty{0},\ty{c_0},\ty{k},\ty{s'+1})$ is \( \Ttemp \)-ordered (possibly $\ty{0} = \ty{c_0}$).
Tree $T_k$ realizes $\ty{k}$ with some inner node $\y_j$ with $j\in[s]$.
Our task is to determine $j$ without knowing $T_k$.
Let $\ty{c_1},\dots,\ty{c_z}$ be the (possibly non-existent, possibly containing $\ty{c_0'}$) remaining children of $\ty{k}$ in $\Ttemp$.
Again, by the induction hypothesis, the subtrees $T_{{c_1}},\dots,T_{{c_z}}$ are polynomial time computable.

Tree $T_{c_0}$ realizes $\ty{c_0}$ with some node $\y_{i-1}$ for $i \in \{2,\dots,s\}$ where $\y_0\in\Y_0$.
Because $T_{c_0}$ is a subtree of $T_k$, this limits the possible realizations of $\y_j$ to $\{\y_i,\dots,\y_k\}$.
Let $\overrightarrow{c_0}$ be the path $(\y_{i-1},\y_{i},\dots,\y_s)$.

Consider the adjacency $\lambda_{c_1}\lambda_k$.
A simple case is that $\ttemp(\lambda_{c_1})=\{\lambda_{c_1}\}\in\unsur(G)$.
Then in the tree $T_k$ the two realizing nodes $\lambda_{c_1},\lambda_{k}$ must be adjacent,
	and we define $\overrightarrow{c_1}=\lambda_{c_1}\lambda_k$.
Note that we define the path with a variable $\ty{k}$ (as named to coincide with the template node since it shares the same variability).
For example $\overrightarrow{z_1}(\y_j)$ is the path contained in the tree $T_k$ (which we aim to construct).

Node $\ttemp(\ty{c_1})$ is part of a chain with terminal $Y_{s'+1}'$ where $\ttemp(\ty{k})\subseteq Y_{s'+1}'$.
Now either $\ty{c_1}$ is the other terminal of chain $\YY_1$, or is the set of inner nodes $\inner(\YY_1)$ and hence realized as one of them.
Thus the format of the chain is either
\begin{itemize}
\item
$\langle Y_0',\{\y_{c_1}^1\},\dots, \{\y_{c_1}^{s_{1}}\}, Y_{s'+1}' \rangle^\gr$ where $\ttemp(\ty{c_1}) \subseteq Y_0'$, or
\item
$\langle \cdot,\dots, \{\y_{c_1}\},\{\y_{c_1}^1\},\dots, \{\y_{c_1}^{s_{1}}\}, Y_{s'+1}' \rangle^\gr$,
 	where $\y_{c_1}$ is the realization of $\ty{c_1}$ in the tree $T_{c_1}$.
\end{itemize}
Let $\overrightarrow{c_1}(\ty{k})$ be the path $(\ty{c_1},\y_{c_1}^1,\dots,\y_{c_1}^{s_{1}},\ty{k})$.
Note that we define the path with a variable $\ty{k}$; with a name coinciding with the template node, since it shares the same variability.
For example $\overrightarrow{c_1}(\y_j)$ is the path contained in the unknown tree $T_k$.
We define the paths $\overrightarrow{c_2}(\ty{k}),\dots, \overrightarrow{c_z}(\ty{k})$ for the other children analogously.
Clearly, the same observations apply.

We similarly define the tree $T(\ty{k})$.
Tree $T(\ty{k})$ is the tree containing the subtrees $T_{c_0},T_{z_1},\dots,T_{c_z}$ together with paths $\allowbreak\overrightarrow{c_1}(\ty{k}),\dots,\overrightarrow{c_z}(\ty{k})$ and path $\y_{i-1},\y_i,\dots,\y_s$.
Namely, $T_k$ is the subtree of $T(\y_j)$ rooted at $\y_j$.
	Similarly as for $T_k'$, we define $T'(\ty{k})$.
	Again there is a node $\ty{k'}$ closest to the global root $\gr_1$,
		and we let $T_k'$ be the tree $T_k$ where the realization $\y$ of $\ty{k'}$ has an additional leaf $\y'$ adjacent to $\y$.
Then $T_k'$ is the subtree of $T'(\y_j)$ rooted at $\y_j$.
Again, considering this tree $T_k'$ with an extra leaf node may be necessary when considering the potential $\pot$ where then vertices may escape at edge $(\y,\y')$.

Now it remains to determine $\y_j$ without knowing $T_k$.
For that purpose, consider the tree $T'(\y_i)$, the tree with the most conservative realization of $\ty{k}$.
Applying the rehang operation yields $ \treeop{\y_i,\y_j}T'(\y_i) = T'(\y_j) $.
Since $\ttarget$ is normalized and because of locality
\[
	\langle\y_j\rangle= \pot(\ttarget,\y_j)
	= \pot(\treeopsss{\y_i} \ttarget, \y_j)
	= \pot(\treeopsss{\y_i} T'(\y_j), \y_j)
	= \pot(T'(\y_j),\y_j) .
\]
Then by the linearity of the potential we have that $\pot(T'(\y_i),\y_i)=\langle\y_i,\dots,\y_j\rangle$.
This is how we algorithmically determine $\y_j$, assuming a yes-instance.
Thus the desired tree $T_k(\y_j)$ is polynomial time computable given graph $G$, template $(\Ttemp,\ttemp,\htemp)$ and $\gr$.

If our algorithm observes that \( \pot(T'(\y_i),\y_i) = \emptyset \) at some point, no normalized representation \( \ttarget \) can exist, and our algorithm may output no.

\medskip

Let us show how to obtain a $\Oh(n^3)$ runtime.
As stated before, no decision making is needed for the leaves and branching nodes $\ty{k}$ where $\ttemp(\ty{k})\in\unsur$.
Though at a branching node $\ty{k}$ where $\ttemp(\ty{k})\in\inner$ we have to determine the potential $\pot(T'(\y_i),\y_i) = \pot(\treeopsss{\y_i}T'(\y_i),\y_i) = \langle \y_i,\dots,\y_j \rangle$ as observed before.
To do this we iterate over the chain from $\y_i$ towards the terminal $Y_{s+1}$.
At some node $\y_{j+1}$ we may observe for the first time that the subtree $\treeopsss{\y_i,\y_j}T'(\y_i)$ contains vertices $u,v$ that falsify the membership of $\y_{j+1}$ in the potential,
	that is $u$ has model $M_u\subseteq V(\treeops{\y_i,\y_j}T'(\y_i))$ and does not escape $v$.
This then implies $\pot(\treeopsss{\y_i}T'(\y_i),\y_i) = \langle \y_i,\dots,\y_j \rangle$.

To check for pairs of vertices $u,v$ that potentially show $\y_{j+1} \notin \pot(\treeopsss{\y_i,\y_j}T'(\y_i),\y_i)$, let us throughout our construction maintain an $\Oh(n^2)$ size table for every node $\y_i$ that marks every ordered pair $(u,v)$ if $u$ does not escape $v$ when we restrict the models to $M_u \cap (\treeops{\y}T'(\y_i) \cup \{\y_i,\y_{i-1}\})$ respectively $M_v \cap (\treeops{\y}T'(\y_i) \cup \{\y_i,\y_{i-1}\})$, where $\y_{i-1}$ is the child of $\y_i$.
We determine this table inductively for the nodes of our constructed subtrees $T_1,\dots,T_k$.
For a node $\y$ with one child $\y'$ we can update every table entry from the table of $\y'$ by considering $\V{\y} \cup \V{\y'}$ and the differences $\V{\y} \setminus \V{\y'}$ and $\V{\y'} \setminus \V{\y}$, in time $\Oh(n^2)$.
	The table of $\y$ marks that $u$ does not escape $v$ if $\y'$ marks that $u$ does not escape $v$ and $v \in \V{\y}$, or if $u \in \V{\y}\setminus\V{\y'}$ and $v \in \V{\y}\cup\V{\y'}$.
If a node $\y$ has children $\y_1,\dots,\y_z$
first determine two tables for $\y$, one as if $\y_1$ were the only child of $\y$, and one as if $\y_2$ were the only child of $\y$, and so on.
Then we combine these tables in time $\Oh(n^2 \cdot z)$ as follows.
We mark that $u$ does not escape $v$ for node $\y$, if every child vertex set $\V{\y_i}, i \in [z]$ contains $v$ and, if it contains also $u$, has a mark for $u$ not escaping $v$.
Note that we can ignore the factor $z$, the number of children, by counting the number of parent nodes instead of child nodes:
	Every non-leaf child node has at most one parent and thus causes at most a constant many additional $\Oh(n^2)$-time operations.
	We may ignore here the children with empty vertex-set, as they are trivial to handle.

While trying different realizations of a branching node, our algorithm notices when $\y_{j+1} \notin \pot(\treeopsss{\y_i}T'(\y_i),\y_i)$,
	because then there is a mark for that $u$ does not escape $v$ and $u \in \V{\y_{j}} \setminus \V{\y_{j+1}}$ while $v \in \V{\y_j}\cup \V{\y_{j+1}}$.
Here we have to potentially try every node of a chain as a branching node.
However once $\y_j$ is determined as a branching node, we no longer have to try the nodes $\y_i,\dots,\y_j$ as a potential branching node; hence no node is tried more than once as a branching node.
Also we may combine the tables of the $z-1$ non-chain children once to a preliminary table, and use this table when trying $\y_i,\dots,\y_j$ as a branching.
Finally since we perform this update step for every of the $\Oh(n)$ non-leaf nodes of $\ttarget$ at most a constant time, we obtain the runtime of $\Oh(n^3)$.
\end{proof}

\newcommand{\commentt}{
\begin{proof}[Proof (Sketch).]
\newcommand{\ttarget}{\Tsub'}
Assuming a yes-instance, there is minimal \compact representation \( \ttarget \) of \( G \) that realizes template $(\Ttemp,\ttemp,\htemp)$.
We may also assume that $\ttarget$ is normalized (proven in the appendix).
Our algorithm outputs a representation isomorphic to $\ttarget$, thus a normalized one as desired.
If, however, our construction fails at some point, we correctly conclude that no such representation exists.
In the rest of the proof we fix an arbitrary root-ordering~$\gr$.

We fix an ordering $\sigma=\ty1,\ty2,\dots$ of the non-leaf nodes of the template tree $\Ttemp$,
	which follows the ordering within in a chain and otherwise is bottom-up.
Pick a node $\ty{k}$ where every non-leaf child of $\ty{k}$ has been added before, and append it to the ordering.
If there is a chain $\langle Y_0,\dots,Y_{s'+1} \rangle^\gr$ mapped by $\htemp$ to a path of form $\ty{0},\ty{k},\ty{k,1},\dots,\ty{k,s'+1}$, append nodes $\ty{k,1},\dots,\ty{k,s'+1}$ as well.
Then continue to picking a new node until all nodes are ordered.

For $k\geq1$, let $T_k$ be the subtree of $\Tsub'$ induced by $\ty{1},\dots,\ty{k}$, every subdivision node between nodes from $\ty{1},\dots,\ty{k}$ and leaves neighboring $\ty{1},\dots,\ty{k}$.
By induction over $k \geq 1$, we prove that a tree isomorphic to $T_k$
	is polynomial time computable given \( G \), $(\Ttemp,\ttemp,\htemp,)$ and $\gr$.
Eventually this yields to a representation \( \Tsub \) isomorphic to \( \ttarget \), thus normalized minimal \compact and realizing $(\Ttemp,\ttemp,\htemp)$, as desired.

(Induction base, when $\ty{k}$ neighbors a leaf (w.r.t.\ to root $\gr_1$))
The node $\ty{k}$ represents a not-surrounded node $\ttemp(\ty{k}) = \{\ty{k}\} \in \unsur(G)$.
Then \( T_k \) consists only of $\ty{k}$ adjacent to some leaf.
Thus, this tree is prescribed by $(\Ttemp,\ttemp,\htemp)$ and hence no computation is required.
The induction step where $\ty{k} \in V(\Ttemp)$ is a node with $\ttemp(\ty{k}) = \{\ty{k}\} \in \unsur(G)$ is similar, and omitted here.

(Induction step $\inner(G)$)
We consider the case where the template node $\ty{k}$ is a surrounded branching node.
This means that $\ttemp(\ty{k}) = \{\y_1,\dots,\y_s\}=\innerr(\YY)$ for some chain $\YY$.
The template maps $\YY$ to a non-trivial path $\ty{0},\dots,\ty{s'+1}$ in $\Ttemp$ containing $\ty{k}$:
\[
\ty{0}\dots\ty{c_0}\ty{k}\ty{c_0'}\dots\ty{s'+1}
\; = \;
\htemp\big(\langle Y_0,\{\y_1\} \dots,\{\y_i\},\dots,\{\y_s\}
,Y_{s+1} \rangle^\gr\big) .
\]
For each of those inner template nodes $\lambda_{i'}$, we have $\ttemp(\lambda_{i'})=\{\y_1,\dots,\y_s\}$.
Let us assume that $\ttemp(\ty{0})\subseteq Y_0$ such that the directions of increasing indices match.

Note that $\ty{c_0}$ is ordered before $\ty{k}$ because of how the ordering $\sigma$ is defined.
Our algorithm may determine $\ty{c_0}$ as the child in $\Ttemp$ where $(\ty{0},\ty{c_0},\ty{k},\ty{s'+1})$ is \( \Ttemp \)-ordered (possibly $\ty{0} = \ty{c_0}$).
The tree $T_k$ realizes $\ty{k}$ with some inner node $\y_j$ with $j\in[s]$.
Our task is to determine $j$ without knowing $T_k$.
Let $\ty{c_1},\dots,\ty{c_z}$ be the (possibly non-existent, possibly containing $\ty{c_0'}$) remaining children of $\ty{k}$ in $\Ttemp$.
By the induction hypothesis, the subtrees $T_{c_0},T_{{c_1}},\dots,T_{{c_z}}$ are polynomial time computable.

The tree $T_{c_0}$ realizes $\ty{c_0}$ with some node $\y_{i-1}$ for $i \in \{2,\dots,s\}$ where $\y_0\in\Y_0$.
Because $T_{c_0}$ is a subtree of $T_k$, this limits the possible realizations of $\y_j$ to $\{\y_i,\dots,\y_s\}$.
Let $\overrightarrow{c_0}$ be the path $(\y_{i-1},\y_{i},\dots,\y_s)$.

Consider the adjacency $\lambda_{c_1}\lambda_k$.
A simple case is that $\ttemp(\lambda_{c_1})=\{\lambda_{c_1}\}\in\unsur(G)$.
Then in the tree $T_k$ the two realizing nodes $\lambda_{c_1}$ and $\lambda_{k}$ must be adjacent,
	and we define $\overrightarrow{c_1}(\lambda_k)$ to be the path $\lambda_{c_1}\lambda_k$.
Note that we define the path with a variable $\ty{k}$ (as named to coincide with the template node since it shares the same variability).
For example $\overrightarrow{c_1}(\y_j)$ is the path contained in the tree $T_k$ (which we aim to construct).

Otherwise, the node $\ttemp(\ty{c_1})$ is part of a chain with terminal $Y_{s'+1}'$ where $\ttemp(\ty{k})\subseteq Y_{s'+1}'$.
Now either $\ty{c_1}$ is the other terminal of chain $\YY_1$, or it is the set of inner nodes $\innerr(\YY_1)$ and hence realized as one of them.
Thus the format of the chain is either
\begin{itemize}
\item
$\langle Y_0',\{\y_{c_1}^1\},\dots, \{\y_{c_1}^{s_{1}}\}, Y_{s'+1}' \rangle^\gr$ where $\ttemp(\ty{c_1}) \subseteq Y_0'$, or
\item
$\langle \cdot,\dots, \{\y_{c_1}\},\{\y_{c_1}^1\},\dots, \{\y_{c_1}^{s_{1}}\}, Y_{s'+1}' \rangle^\gr$,
 	where $\y_{c_1}$ is the realization of $\ty{c_1}$ in the tree $T_{c_1}$.
\end{itemize}
Let $\overrightarrow{c_1}(\ty{k})$ be the path $(\ty{c_1},\y_{c_1}^1,\dots,\y_{c_1}^{s_{1}},\ty{k})$,
	similarly as before with variable $\ty{k}$.
For example $\overrightarrow{c_1}(\y_j)$ is the path contained in the unknown tree $T_k$.
We define the paths $\overrightarrow{c_2}(\ty{k}),\dots, \overrightarrow{c_z}(\ty{k})$ for the other children analogously.
Clearly, the same observations apply.

We define the tree $T(\ty{k})$ similarly.
Namely, $T(\ty{k})$ is the tree containing the subtrees $T_{c_0},T_{c_1},\dots,T_{c_z}$ together with paths $\allowbreak\overrightarrow{c_1}(\ty{k}),\dots,\overrightarrow{c_z}(\ty{k})$ and path $\y_{i-1},\y_i,\dots,\y_s$.
Then $T_k$ is the subtree of $T(\y_j)$ rooted at $\y_j$.
Thus it remains to determine $\y_j$ without knowing $T_k$.

For that purpose, consider the tree $T(\y_i)$, the tree with the most conservative realization of $\ty{k}$.
Applying the rehang operation yields $ \treeop{\y_i,\y_j}T(\y_i) = T(\y_j) $.
	Assume that node $\y_k$ of all the nodes of $T_k$ has the smallest distance to the global root $\gr$ (the general case is handled by a slight modification to $T(\y_i)$, see appendix).
Then, since $\ttarget$ is normalized and because of locality, we have $ \langle\y_j\rangle= \pot(\ttarget,\y_j)= \pot(\treeopsss{\y_i} \ttarget, \y_j) = \pot(\treeopsss{\y_i} T(\y_j), \y_j) = \pot(T(\y_j),\y_j) $.

Then by the linearity of the potential we have that $\pot(T(\y_i),\y_i)=\langle\y_i,\dots,\y_j\rangle$.
This is how we algorithmically determine $\y_j$, assuming a yes-instance.
Thus the desired tree $T_k(\y_j)$ is polynomial time computable given graph $G$, template $(\Ttemp,\ttemp,\htemp)$ and $\gr$.
If our algorithm observes that \( \pot(T(\y_i),\y_i) = \emptyset \) at some point, it contradicts the existence of a normalized representation \( \ttarget \), and our algorithm returns no.
\end{proof}
}

Now we outline our \cFPT algorithm for the parameter $t = |V(T)|$.
We assume without loss of generality that $G$ is a chordal graph and $T\neq K_1$ as the problem is trivial otherwise.
Note that chordality can be tested in linear time~\cite{RoseTL1976}.
If $G$ is not connected, each proper interval graph component always be represented using a subdivision of an edge incident to a leaf of $T$. Thus,
these components, which can be recognized in linear time~\cite{Corneil04,DengHH1996}, can be excluded from the further consideration.
Each of the remaining components is not a proper interval graph and, as such, contains a vertex whose model includes a branching node of $T$.
Thus, if these components number more than the number of branching nodes, $G$ has no $T$-representation.
Assume that this is not the case.
We guess an assignment of the connected components of~$G$ to connected subtrees of~$T$ representing them.
Two such subtrees may share an edge (which can be needed to represent both components of $G$ using the end-nodes of this shared edge).
Note that are at most $2^{\Oh(t \log t)}$ possible mappings, and then we can deal with every component of $G$ and the corresponding subtree of $T$ separately.
From  now on, we assume that $G$ is connected.
By \autoref{lemma:compact:representation}, we may look for a \compact representation $\Tsub$;
	further, it suffices that $\Tsub$ is minimal.
Therefore, there is a template $(\Ttemp,\ttemp,\htemp)$ that allows a representation of $G$ as seen in \autoref{lemma:template:exists}.
We compute the chains of $G$ and try every template in time $2^{\Oh(t^2 \log t)}\cdot n^3$ where $n= |V(G)|$, as seen in \autoref{lemma:try:templates}.
Pick an arbitrary root-ordering $\gr$.
Then test in polynomial time whether a minimal \compact representation of \( G \) realizing this template by using \autoref{lemma:realization:algorithm}.
In a positive case, applying \autoref{lemma:compact:representation} leads to a proper representation.
This implies our main result (restated here).

\algorithm*

%% file: figure-surround.tex

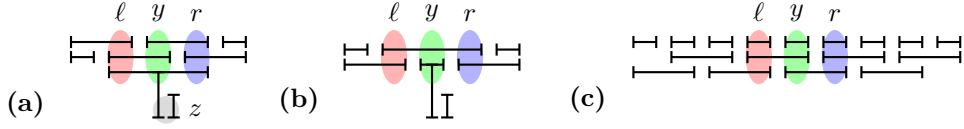
\begin{figure}[t]
\centering
\begin{tikzpicture}[scale=1]
	\begin{scope}[local bounding box=boxA]
		\iiclique{\lcolor}{1.75}{$\ell$}
		\iiclique{\ycolor}{2.25}{$\y$}
		\iiclique{\rcolor}{2.75}{$r$}

		\node[ellipse,draw=\zcolor,fill=\zcolor,label=right:$z$,minimum height=10] at (2.35,0.5) {};

		\intervalll{1}{2}{1.4}
		\intervalll{2}{3}{1.4}
		\intervalll{3}{3.5}{1.4}
		\intervalll{1}{1.5}{1.2}
		\intervalll{1.5}{2.5}{1.2}
		\intervalll{2.5}{3.5}{1.2}
		\intervalll{1.5}{3}{1}

		\vintervall{1.0}{0.4}{2.25}
		\vintervall{0.7}{0.4}{2.45}
	\end{scope}
	\node[xshift=-.6cm, yshift=-.6cm] at (boxA.west) {\textbf{(a)}};
	\begin{scope}[local bounding box=boxB,xshift=3.6cm]
		\iiclique{\lcolor}{1.75}{$\ell$}
		\iiclique{\ycolor}{2.25}{$\y$}
		\iiclique{\rcolor}{2.75}{$r$}

		\intervalll{1}{1.5}{1.3}
		\intervalll{1.5}{3}{1.3}
		\intervalll{3}{3.5}{1.3}
		\intervalll{1}{2}{1.1}
		\intervalll{2}{2.5}{1.1}
		\intervalll{2.5}{3.5}{1.1}

		\vintervall{1.1}{0.4}{2.25}
		\vintervall{0.7}{0.4}{2.45}
	\end{scope}
	\node[xshift=-.6cm, yshift=-.6cm] at (boxB.west) {\textbf{(b)}};
	\begin{scope}[local bounding box=boxC, xshift=8.4cm]
		\iiclique{\lcolor}{1.75}{$\ell$}
		\iiclique{\ycolor}{2.25}{$\y$}
		\iiclique{\rcolor}{2.75}{$r$}

		\foreach \x in {0,1,2,3} {
			\intervalll{\x}{\x+1}{1}
			\intervalll{\x+0.5}{\x+1.5}{1.2}
			\intervalll{\x}{\x+0.5}{1.4}
			\intervalll{\x+0.5}{\x+1}{1.4}
		}\foreach \x in {4} {
		\intervalll{\x}{\x+0.5}{1.4}}
	\end{scope}
	\node[xshift=-.6cm,yshift=-.8cm] at (boxC.west) {\textbf{(c)}};
\end{tikzpicture}
\caption{
(a)
A proper $K_{1,3}$-graph.
Triple $(\ell,\y,r)$ is surrounding.
Any representation positions $y$ between $\ell$ and $r$.
Component $\V{z}\setminus\V{\ell}$ complies with condition~\ref{it:surround:subset}.
Further, edge $\{z,\y\}$ may be replaced by edge $\{z,\ell\}$ or $\{z,r\}$.
(b)
A proper $K_{1,3}$-graph.
Triple $(\ell,\y,r)$ is not surrounding.
A `private' vertex in $\V{\y}\setminus N(\K_\ell)\cup N(\K_r)$ contradicts condition~\ref{it:surround:V:y}.
Indeed, any of $\ell,\y,r$ may realize the branching node.
(c)
Triple $(\ell,\y,r)$ is surrounding.
For $\{\K_\ell,\K_r\}=\KK{\y}$ condition~\ref{it:surround:2} allows private vertices in $\V{\y}$;
otherwise, this proper interval graph would have no surrounded nodes.
}
\label{figure:examples}
\end{figure}

%% file: figure-components.tex
\begin{figure}[t]
\captionsetup{position=b}
\begin{tikzpicture}[scale=1]
\begin{scope}[local bounding box=boxA]
\iiclique{\ycolor}{2.25}{$\y$}

\node[ellipse,draw=\zcolor,fill=\zcolor,minimum height=65, minimum width=25] at (0.55,1.2) {};
\node[ellipse,draw=\zcolor,fill=\zcolor,label=above:$~~M_{\K'}$,minimum height=20, minimum width=40] at (1.2,1.2) {};

\node[ellipse,draw=\zcolor,fill=\zcolor,minimum height=27, minimum width=20] at (2.8,0.7) {};
\node[label=left:$M_{\K''}$]  at (2.8,0.3) {};
\node[ellipse,draw=\zcolor,fill=\zcolor,minimum height=17, minimum width=30] at (3,1.2) {};

\node[ellipse,draw=\zcolor,fill=\zcolor,label=right:$M_{\K}$,minimum height=12, minimum width=12] at (2.85,1.85) {};

\vintervall{0.4}{0.7}{0.55}
\vintervall{0.9}{2}{0.55}

\vintervall{0.4}{1.5}{0.35}
\vintervall{1.7}{2}{0.35}

\vintervall{0.4}{2}{0.75}\intervall{0.75}{1.4}{1.1}

\intervalll{1}{2}{1.3}
\intervalll{1.5}{2.5}{1.1}

\intervallca{2.1}{2.75}{1.3}{black}\vintervallca{2}{1.3}{2.75}{black}

\intervalll{2.5}{3.5}{1.1}\vintervall{1.1}{0.4}{2.75}

\intervalll{3}{3.5}{1.3}
\vintervall{0.4}{0.7}{2.95}
\vintervall{1.7}{2}{2.95}

	\end{scope}
	\node[xshift=-.2cm, yshift=-0.95cm] at (boxA.west) {\textbf{(a)}};
	\begin{scope}[local bounding box=boxB,xshift=4.8cm]
	
\iiclique{\ycolor}{1.25}{$\y$}
\iiclique{\rcolor}{1.75}{$R$}
\node[ellipse,draw=\lcolor,fill=\lcolor,minimum height=65, minimum width=25] at (0.55,1.2) {};
\node[label=right:$L$] at (0.7,0.4) {};

\vintervall{0.4}{0.7}{0.55}
\vintervall{0.9}{2}{0.55}

\vintervall{0.4}{1.5}{0.35}
\vintervall{1.7}{2}{0.35}

\vintervall{0.4}{2}{0.75}\intervall{0.75}{1.4}{1.1}

\intervalll{1}{2}{1.3}
\intervalll{1.5}{2.5}{1.1}

\intervallca{2.1}{2.75}{1.3}{black}\vintervallca{2}{1.3}{2.75}{black}

\intervalll{2.5}{3.5}{1.1}\vintervall{1.1}{0.4}{2.75}

\intervalll{3}{3.5}{1.3}
\vintervall{0.4}{0.7}{2.95}
\vintervall{1.7}{2}{2.95}

	\end{scope}
	\node[xshift=-.2cm, yshift=-0.95cm] at (boxB.west) {\textbf{(b)}};
	\begin{scope}[local bounding box=boxC, xshift=9.3cm]

\iiclique{\lcolor}{1.75}{$\ell$}
\iiclique{\ycolor}{2.25}{$\y$}
\iiclique{\rcolor}{2.75}{$z$}

\node[ellipse,draw=\zcolor,fill=\zcolor,label=right:$M_{\K}$,minimum height=12, minimum width=12] at (3.35,1.85) {};

\vintervall{0.4}{0.7}{0.55}
\vintervall{0.9}{2}{0.55}

\vintervall{0.4}{1.5}{0.35}
\vintervall{1.7}{2}{0.35}

\vintervall{0.4}{2}{0.75}\intervall{0.75}{1.4}{1.1}

\intervalll{1}{2}{1.3}
\intervalll{1.5}{2.5}{1.1}

\intervallca{2.1}{3.25}{1.3}{black}\vintervallca{2}{1.3}{3.25}{black}

\intervalll{2.5}{4}{1.1}\vintervall{1.1}{0.4}{3.25}

\intervalll{3.5}{4}{1.3}
\vintervall{0.4}{0.7}{3.45}
\vintervall{1.7}{2}{3.45}

	\end{scope}
	\node[xshift=-.4cm, yshift=-1cm] at (boxC.west) {\textbf{(c)}};
\end{tikzpicture}
\caption{
(a)
Graph $G-\V{\y}$ splits into three connected components.
Their models partition the non-leaves of $\Tsub-\{\y\}$.
Each model contains an \eye (of this subdivision of a double-star).
(b)~%
Sets $L,R$ exactly capture the nodes $\ell,r$ that (together with~$\y$) form a surrounding triple $(\ell,\y,r)$.
Here, $\y$ neighbors a subdivision node $r$ to its right.
In such a case a $\y$-guard $R\supseteq \{r\}$ may only contain $r$.
(c)
Subdivision node~$\y$ is not surrounded.
A remote component $\K$ falsifies condition~\ref{it:surround:subset}.
Note that $\KKI{\K}=\{\y,z\}$. 
Thus $\K$ does not falsify \ref{it:surround:subset} for another subdivision node $\y'$.
}
\label{figure:partition}
\label{figure:L:R}
\label{figure:falsify:only:one}
\end{figure}
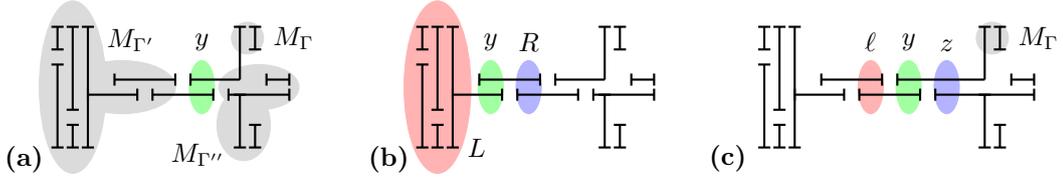

%% file: figure-template.tex
\begin{tikzpicture}[scale=1]
\begin{scope}[local bounding box=boxA]

\node[ellipse,draw=\ycolor,fill=\ycolor,minimum height=20, minimum width=30] at (-0.8,2.3) {};
\node[ellipse,draw=\lcolor,fill=\lcolor,minimum height=20, minimum width=30] at (-0.8,0.5) {};

\newcommand{\shift}{1.2}
\vintervall{0.4}{0.7}{0.55-\shift}
\vintervall{0.9}{2.4}{0.55-\shift}

\vintervall{0.4}{1.5}{0.35-\shift}
\vintervall{1.7}{2.4}{0.35-\shift}

\vintervall{0.4}{1.9}{0.15-\shift}
\vintervall{2.1}{2.4}{0.15-\shift}

\vintervall{0.4}{2.4}{0.75-\shift}

	\node[label=left:$\YY$] at (0.2-\shift,1.3) {};


\node[ellipse,draw=\zcolor,fill=\zcolor,minimum height=70, minimum width=30] at (0.45,1.4) {};
\iiclique{\rcolor}{2.25}{}

\vintervall{0.4}{0.7}{0.55}
\vintervall{0.9}{2.4}{0.55}

\vintervall{0.4}{1.5}{0.35}
\vintervall{1.7}{2.4}{0.35}

\vintervall{0.4}{1.9}{0.15}
\vintervall{2.1}{2.4}{0.15}

\vintervall{0.4}{2.4}{0.75}\intervall{0.75}{1.4}{1.1}

\intervalll{1}{2}{1.3}
\intervalll{1.5}{2.5}{1.1}
\intervalll{2}{2.5}{1.3}
	\node[label=above:$\YY'$] at (1.55,1.2) {};

%
%
%

	\end{scope}
	\node[xshift=-.1cm, yshift=-1.1cm] at (boxA.west) {\textbf{(a)}};
	\begin{scope}[local bounding box=boxB,xshift=3.7cm]

\draw (0,0.1) rectangle (0.9,0.2) [];
\draw (0,0.3) rectangle (0.9,0.75) [fill=\lcolor];
	\node[label=left:$\ty{0}$] at (0.2,0.5) {};
\draw (0,0.85) rectangle (0.9,1.95);
	\node[label=left:$\ty{1}$] at (0.2,1.3725) {};
\draw (0,2.05) rectangle (0.9,2.5) [fill=\ycolor];
	\node[label=left:$\ty{2}$] at (0.2,2.25) {};
\draw (0,2.6) rectangle (0.9,2.7) [];

\draw (2,1) rectangle (2.45,1.4) [fill=\rcolor];
	\node[label=above:$\ty{}$] at (2.25,1.2) {};
\draw (2.55,1) rectangle (3.4,1.4) [fill=\zcolor];
\draw (3.5,1) rectangle (4,1.4) [fill=\zcolor];
\draw (4.1,1) rectangle (4.2,1.4) [];

\draw (2.9,1.5) rectangle (3.4,2.1) [fill=\zcolor];
\draw (2.9,2.2) rectangle (3.4,2.3) [];

\draw (2.9,0.3) rectangle (3.4,0.9) [fill=\zcolor];
\draw (2.9,0.1) rectangle (3.4,0.2) [];

\vintervall{0.4}{0.7}{0.55}
\vintervall{0.9}{2.4}{0.55}

\vintervall{0.4}{1.5}{0.35}
\vintervall{1.7}{2.4}{0.35}

\vintervall{0.4}{1.9}{0.15}
\vintervall{2.1}{2.4}{0.15}

\vintervall{0.4}{2.4}{0.75}\intervall{0.75}{1.4}{1.1}

\intervalll{1}{2}{1.3}
\intervalll{1.5}{2.5}{1.1}

\intervallca{2.1}{3.05}{1.3}{black}\vintervallca{2}{1.3}{3.05}{black}

\intervalll{2.5}{4}{1.1}\vintervall{1.1}{0.4}{3.05}

\intervalll{3.5}{4}{1.3}
\vintervall{0.4}{0.7}{3.25}
\vintervall{1.7}{2}{3.25}

	\end{scope}
	\node[xshift=-.25cm, yshift=-1.1cm] at (boxB.west) {\textbf{(b)}};
	\begin{scope}[local bounding box=boxC, xshift=10.3cm]

\newcommand{\rheight}{2.2}
\node[ellipse,draw=\rcolor,fill=\rcolor,minimum height=15, minimum width=10] at (-0.50,\rheight) {};
\node[ellipse,draw=\ycolor,fill=\ycolor,minimum height=15, minimum width=10] at (0.00, \rheight) {};
\node[ellipse,draw=\lcolor,fill=\lcolor,minimum height=15, minimum width=10] at (0.5, \rheight) {};
\node[] at (0,\rheight) {$\gr=\ty{},\ty{2},\ty{0}\dots$};

\renewcommand{\rheight}{1.4}
\node[ellipse,draw=\rcolor,fill=\rcolor,minimum height=20, minimum width=20] at (0.5,\rheight) {};
\node[] at (0,\rheight) {$\big\langle Y_0',\dots,\{\y_2'\} \big\rangle^\gr $}; 

\renewcommand{\rheight}{0.6}
\node[ellipse,draw=\lcolor,fill=\lcolor,minimum height=15, minimum width=10] at (-0.65, \rheight) {};
\node[ellipse,draw=\ycolor,fill=\ycolor,minimum height=15, minimum width=10] at (0.5, \rheight) {};
\node[] at (0,\rheight) {$\big\langle Y_0,\dots,Y_3 \big\rangle^\gr $};

	\end{scope}
	\node[xshift=-.2cm, yshift=-1.1cm] at (boxC.west) {\textbf{(c)}};
\end{tikzpicture}
\caption{
a)
Depiction of the chains of almost the graph from \autoref{figure:partition}a) with marked terminals.
b)
A template prescribes the positioning of terminals, chains and $\unsur$.
Empty boxes represent leaves.
The chain $\YY$ is mapped to path $\ty{0}\ty{1}\ty{2}$.
The chain $\YY'$ is mapped to the single edge path $\ty{1}\ty{}$ (an edge despite the gap in the picture).
We have $\ttemp(\lambda_1)=\innerr(\YY)$, and indeed any attachment of $\YY'$ to an inner node of $\YY$ is possible.
For the chain $\YY'$ a mapping to single edge suffices since here it is realized by a path of subdivision nodes.
c)
A sample root-ordering (colors mark the mapping of $\ttemp$),
and resulting orientation of chains $\YY'$ and $\YY$.
Here $\ty{2}$ (mapped to $Y_3$) is a tie-breaker for $\YY$.
}

%% file: recognition-hardness.tex
\section{Recognition Hardness}
\label{sec:rec-hard}

In this section, we discuss the \cNP-completeness of the natural problems of deciding, for a given graph $G$ and a given or fixed graph $H$, whether $G$ is a proper $H$-graph.
As discussed in the introduction, if $H$ is part of the input and we drop the ``proper'' part, it is known that this decision problem is \cNP-complete even when $H$ is restricted to being a tree~\cite{KKOS15}. 
This proof also applies to the case of recognizing proper $H$-graphs when $H$ is part of the input. 
With this in mind we turn to the case when the graph $H$ is fixed.  
Again, when dropping the ``proper'' condition, it is known that $H$-graph recognition is \cNP-complete whenever $H$ is not a cactus or, equivalently, when $H$ contains the diamond graph as a minor~\cite{ChaplickTVZ17}. 
This hardness proof does not directly carry over to the proper case, but we will show that it can be adapted to show a similar hardness result for proper $H$-graphs. 
Namely, we establish \cref{thm:rec-hard-intro}, restated next. 

\recHard*

\begin{figure}[b]
\centering
\includegraphics{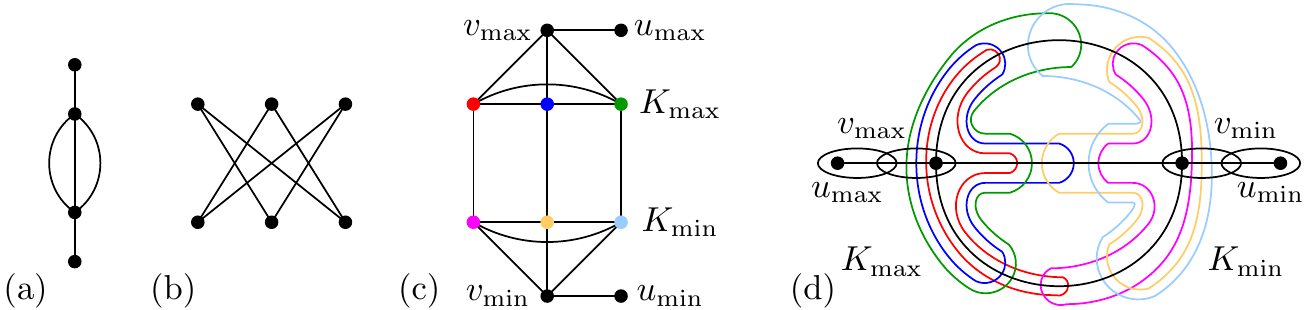}
\caption{(a) A graph $\Dplus$ for which proper $\Dplus$-graph recognition is \cNP-complete. (b) The Hasse diagram of a poset $P$. (c) The graph $G$, which is the incomparability graph of $P$ together with $4$ extra vertices. (d) A $\Dplus$-representation of $G$.}
\label{fig:diamond-plus}
\end{figure}

Prior to providing the proof, we introduce the problem (denoted \hid) used in our reduction.
Let $I$ be a collection of closed intervals on the real line. 
A poset $\calP_I=(I,<)$ can be defined on $I$ by considering intervals $x,y \in I$ and setting $x < y$ if and only if the right endpoint of $x$ is strictly to the left of the left endpoint of $y$. 
A partial order $\calP$ is called an \emph{interval order} when there is an $I$ such that $\calP = \calP_I$. 
The \emph{interval dimension} of a poset $\calP=(P,<)$, is the minimum number of interval orders whose intersection is $\calP$, i.e., for elements $x,y\in P$, $x < y$ if and only if $x$ is before $y$ in all of the interval orders. 
The problem \hid is testing if a given \emph{height one} poset has interval dimension at most three and was shown to be \cNP-complete by Yannakakis~\cite{yannakakis1982}. 

We introduce a bit more background and terminology before presenting the proof. 
The \emph{incomparability graph} $G_\calP$ of a poset $\calP=(P,<)$  is the graph with $V(G) = P$ and $uv \in E(G_\calP)$  if and only if $u$ and $v$ are not comparable in $\calP$.
Note that if $\calP$ has height one, every element is either minimal or maximal and, as such, $G_\calP$ is the \emph{complement of a bipartite graph}. 
Thus, the vertices $V(G_\calP)$ naturally partition into two cliques $K_{\max}$ and $K_{\min}$ containing the maximal and the minimal elements of $\calP$, respectively. 
With these definitions in place we now prove the theorem of this section. 
The main idea, just as in~\cite{ChaplickTVZ17}, is that the three parallel edges will correspond to the three interval orders and will certify that the poset has interval dimension at most 3. 
The key new aspects of this proof in comparison with~\cite{ChaplickTVZ17} are that we have to take extra care to ensure that the graph we construct has a proper representation and to extract the three interval orders from any proper representation of our constructed graph.

\begin{proof}[Proof of \cref{thm:rec-hard-intro}]
We use rely on the \cNP-hardness of \hid. 
Namely, from a given poset $\calP$ we construct a graph $G$ where $G$ is a $\Dplus$-graph if and only if $\calP$ is a yes-instance of \hid. 
Essentially, the three paths connecting the two degree~3 nodes in $\Dplus$ will encode the three
interval orders whose intersection is $\calP$. 
To construct $G$, we start from the incomparability graph $G_{\calP}$ and add four vertices $u_{\min}, v_{\min}$ and $u_{\max}, v_{\max}$ where the neighborhood of $v_{\min}$ is $K_{\min} \cup \{u_{\min}\}$ and the neighborhood of $v_{\max}$ is $K_{\max} \cup \{u_{\max}\}$. 

There are two differences between this graph $G$ and the one used in~\cite[Theorem 10]{ChaplickTVZ17}. First, we only add one vertex to each of $K_{\min}$ and $K_{\max}$ instead of a special 7-vertex tree. Second, we add vertices $u_{\min}$ and $u_{\max}$ which are only adjacent to $v_{\min}$ and $v_{\max}$ (respectively) whereas in~\cite{ChaplickTVZ17}, no such vertices are added. 
We will refer to properties easily gained from~\cite[Theorem 10]{ChaplickTVZ17} as property $(*)$ and use it to avoid unnecessarily repeating their arguments.
Let $a_{\min}$ and $a_{\max}$ be the two degree~4 nodes of $\Dplus$, and let $b_{\min}$ and $b_{\max}$ be their respective degree~1 neighbors. 
The rest of the proof concerns the two natural directions; see Figure~\ref{fig:diamond-plus}.

\subparagraph{Part 1:} Suppose that $\calP$ is a yes-instance of \hid, and let $I_1, I_2, I_3$ be the three interval orders certifying this. 
By $(*)$, the subgraph $G \setminus \{u_{\min},u_{\max}\}$ can be represented by $\{D_v : v \in V(G \setminus \{u_{\min},u_{\max}\})\}$ on a subdivision of the three parallel edges of $\Dplus$ (i.e., on the diamond subgraph of $\Dplus$) with the following properties: 
\begin{itemize}
\item the model of $v_{\min}$ is simply the node $a_{\min}$ and the model of $v_{\max}$ is $a_{\max}$, and
\item for any distinct elements $u,v$ of $\calP$, $D_u \neq D_v$. 
\end{itemize}
In particular, it only remains to modify these models $\{D_v : v \in V(G)\}$ to ensure that they are proper and to create models for $u_{\min}$ and $u_{\max}$. 
This is accomplished as follows via the two thus far unused edges in $\Dplus$, i.e., $a_{\min}b_{\min}$ and $a_{\max}b_{\max}$. 

Note that, for $u \in K_{\min} \cup \{v_{\min}\}$ and $v \in K_{\max} \cup \{v_{\max}\}$, $D_u$ is not contained in $D_v$ since $D_u$ intersects $D_{v_{\min}}$ but $D_v$ does not (symmetrically $D_v$ not contained in $D_u$ due to $v_{\max}$). 
Thus, it suffices for us to modify the models of $K_{\min} \cup \{v_{\min}\}$ so that they become proper (and do this in such a way that it does not affect their relationship with those in $K_{\max} \cup \{v_{\max}\}$. 
We only describe this modification for the elements of $K_{\min} \cup \{v_{\min}\}$ as it is symmetric for $K_{\max} \cup \{v_{\max}\}$. 

Let $\calP_{\min} = (K_{\min} \cup \{v_{\min}\}, <)$ be a poset where, for elements $u,v$, we have $u < v$ if and only if $D_u$ is a proper subset of $D_v$. 
Now consider any linear extension $\pi = (v_1, \ldots, v_\ell)$ of $\calP_{\min}$, and subdivide the $b_{\min}a_{\min}$ edge of $\Dplus$ into a path $b_{\min}=b_0, b_1, b_2, \ldots, b_\ell, b_{\ell+1} = a_{\min}$. 
For each $v_i \in K_{\min} \cup \{v_{min}\}$ we define $\Dplus_v$ as the union of $D_v$ and the path $b_i, \ldots, b_{\ell+1} = a_{\min}$. 
It is easy to see that this results in a proper $\Dplus$-representation for the elements of $K_{\min} \cup \{v_{\min}\}$ and that we have not altered the relationship between models of elements of $K_{\min} \cup \{v_{\min}\}$ and models of elements of $K_{\max} \cup \{v_{\max}\}$. 
Finally, we set $\Dplus_{u_{\min}}$ to be the edge $b_{\min}b_{1}$, i.e.,  $\Dplus_{u_{\min}}$ only intersects $\Dplus_{v_{\min}}$ and $\Dplus_{u_{\min}}$ is not contained $\Dplus_{v_{\min}}$ since $\Dplus_{v_{\min}}$ is the path $b_1, \ldots, b_{\ell+1}$. 
Thus, by symmetrically applying this construction for the vertices of $K_{\max} \cup \{v_{\max},u_{\max}\}$, our graph $G$ is a proper $\Dplus$-graph as claimed.

\subparagraph{Part 2:}
Suppose that $G$ is a proper $\Dplus$-graph and let $\{\Dplus_v : v \in V(G)\}$ be a corresponding representation on a subdivision $\Dplus'$ of $\Dplus$. 
We will modify this proper $\Dplus$-representation into a $\Dplus$-representation where we can easily extract the needed three interval orders. 

We start by adjusting the models of $u_{\min}$, $v_{\min}$, $u_{\max}$, $v_{\max}$, so that: 
\begin{itemize}
\item $\Dplus_{u_{\min}}$ and $\Dplus_{u_{\max}}$ each contain exactly two nodes of $\Dplus'$---the nodes in $\Dplus_{u_{\min}}$ are referred to as $x_{\min},y_{\min}$ and those in $\Dplus_{u_{\max}}$ are referred to as $x_{\max},y_{\max}$; and
\item each of $x_{\min},y_{\min}$ $x_{\max},y_{\max}$ have degree two in $\Dplus'$. 
\end{itemize}
First, observe that, since $v_{\min}$ is the only neighbor of $u_{\min}$, the nodes of $\Dplus_{u_{\min}}$ partition into two non-empty sets $X$ and $Y$, where $X$ contains the nodes which are also contained in $\Dplus_{v_{\min}}$ and $Y$ contains the nodes which are only contained in $\Dplus_{u_{\min}}$. 
Clearly, there must be an $x \in X$ and a $y \in Y$ such that $x$ and $y$ are adjacent in $\Dplus'$. 
Now we subdivide the edge $xy$ to become a path $x,x_{\min},y_{\min},y$ and replace $\Dplus_{u_{\min}}$ by the edge $x_{\min}y_{\min}$ (and similarly replace $\Dplus_{u_{\max}}$). 
Our models now satisfy the above conditions. We now label the degree 3 node of $\Dplus$ which is closest to $x_{\min}$ as $a_{\min}$, and the other degree 3 node as $a_{\max}$. 

Note that, since $y_{\min}$ and $y_{\max}$ each have degree two, they must be subdivision nodes of $\Dplus'$. 
Furthermore, the edges $y_{\min}z$ ($z \neq x_{\min}$) and $y_{\max}z'$ ($z' \neq x_{\max}$) are not contained in the model of any vertex of $G$. 
In particular, we have a proper representation of $G$ on $\Dplus'' = \Dplus' \setminus \{y_{\min}z, y_{\max}z'\}$. 
Observe that, in $\Dplus''$, we have at most three distinct paths which start at $x_{\min}$ and end at $x_{\max}$. Moreover, by restricting the models of the vertices of $K_{\min} \cup K_{\max}$ to the nodes on these three paths, we still have a $\Dplus''$ representation of $G$, i.e., if a pair of models intersect, then they must intersect on at least one of these paths. 
Let these (at most) three paths be $P_1, P_2, P_3$ (these will become our three interval orders). 
In particular, if, for each element $p$ of $K_{\min}$, $\Dplus_{p} \cap P_i$ is a prefix of $P_i$ and (symmetrically) for each element $q$ of $K_{\max}$, $\Dplus_{p} \cap P_i$ is a suffix of $P_i$, then 
these three paths would naturally provide our desired three interval orders. 
We refer to this as property $(P)$. 
So, to complete the proof, it suffices to modify the current proper $\Dplus''$-representation to be a $\Dplus''$-representation satisfying property $(P)$. 

We describe this transformation for each element of $K_{\min}$ and it can be performed symmetrically for the elements of $K_{\max}$. 
For each $p \in K_{\min}$, we replace $\Dplus_{p}$ by $\Dplus_{p} \cup (\Dplus_{v_{\min}}) \setminus \{y_{\min}\}$. 
It is easy to see that this preserves the property of being a $\Dplus$-representation of our graph $G$. 
Moreover, this ensures that $\Dplus_{p} \cap P_i$ contains a prefix of $P_i$ (for each $i$). 

Now, suppose that $\Dplus_{p} \cap P_i$ is not a path. 
Clearly, since $\Dplus$ only contains two degree 3 nodes and since $\Dplus_p$ is connected, $\Dplus_p \cap P_i$ must contain both $a_{\min}$ and $a_{\max}$, i.e., $P_i \setminus \Dplus_{p}$ consists of precisely two paths. 
Let $c_1, \ldots, c_t$ be the first path in $P_i \setminus \Dplus_{p}$, i.e., the connected component which does not contain $x_{\max}$. 
Additionally, let $c_0$ ($\neq c_2$) be the neighbor of $c_1$ on $P_i$ and let $c_{t+1}$ ($\neq c_{t-1}$) be the neighbor of $c_t$.
Consider an element $q$ of $K_{\max}$ where $\Dplus_{q}$ contains some $c_j$ for $j \in [1,t]$. 
Observe that $\Dplus_q$ must also contain either $c_0$ or $c_{t+1}$, i.e., $\Dplus_q \cap \Dplus_q \neq \emptyset$. 
Thus, we can simply add the path $c_0, c_1, \ldots, c_t, c_{t+1}$ to $\Dplus_p$ without introducing new intersections with the models of the elements of $K_{\max}$, i.e., by doing this, $\Dplus_p \cap P_i$ becomes a prefix of $P_i$ as needed. 

Therefore, by applying these modifications to each $p \in K_{\min}$ and (analogously) to each $q \in K_{\max}$, we obtain a $\Dplus''$-representation of $G$ satisfying $(P)$ (completing the proof). 
\end{proof}